\documentclass[
a4paper,
  onecolumn,
  11pt, 
  draftcls,
 journal
]{IEEEtran}

\usepackage{array}
\usepackage{epsfig}
\usepackage{cite}
\usepackage{amsfonts}
\usepackage{amsmath}
\usepackage{amssymb}
\usepackage{amsthm}
\usepackage{amsxtra}
\usepackage{varioref}
\usepackage{multicol}
\usepackage{float}
\usepackage{rotate}
\usepackage{color}
\usepackage{enumerate}
\usepackage[active]{srcltx} 
\usepackage{mathrsfs}
\usepackage{rotating}
\usepackage{Milancommands}
\usepackage{subfigure}

\usepackage{times}
\usepackage[english]{babel}
\usepackage[utf8]{inputenc}
\usepackage{fontenc}
\usepackage[colorlinks=true,linkcolor=blue,citecolor=green]{hyperref}

\makeatletter
\newcommand{\subalign}[1]{%
  \vcenter{%
    \Let@ \restore@math@cr \default@tag
    \baselineskip\fontdimen10 \scriptfont\tw@
    \advance\baselineskip\fontdimen12 \scriptfont\tw@
    \lineskip\thr@@\fontdimen8 \scriptfont\thr@@
    \lineskiplimit\lineskip
    \ialign{\hfil$\m@th\scriptstyle##$&$\m@th\scriptstyle{}##$\crcr
      #1\crcr
    }%
  }
}
\makeatother

\renewcommand{\procx}{\rvax_{1}^{\infty}}
\renewcommand{\procy}{\rvay_{1}^{\infty}}

\newtheorem{thm}{Theorem}
\newtheorem{defn}{Definition}
\newtheorem{coro}{Corollary}
\newtheorem{lem}{Lemma}
\newtheorem{rem}{Remark}
\newtheorem{prop}{Proposition}

\newtheorem{fact}{\textbf{Fact}}
\newtheorem{property}{\textbf{P\!\!}}
\newtheorem{assu}{\textbf{Assumption}}

\title{
Stationarity in the Realizations of the \\Causal Rate-Distortion Function for \\One-Sided Stationary Sources
} 

\author{Milan S. Derpich\thanks{Milan S. Derpich and Marco A. Guerrero are with the Department of Electronic Engineering, Universidad T\'ecnica Federico Santa Mar\'ia, Av. Espa\~na 1680, Valpara\'iso, Chile. Their work was partially funded by CONICYT grants FONDECYT 1171059 and FB0008.}, 
Marco A. Guerrero and 
Jan {\O}stergaard\thanks{Jan {\O}stergaard is with the Department of Electronic Systems, 
Aalborg University, Fredrik Bajers Vej 7, DK-9220, Denmark, jo@es.aau.dk
}
}

\begin{document}

\maketitle
\begin{abstract} 
This paper derives novel results on the characterization of the
the causal \textit{information rate-distortion function} (IRDF) 
$R_{c}^{it}(D)$
for arbitrarily-distributed one-sided stationary $\kappa$-th order Markov source $\rvax_{1}^{\infty}=\rvax(1),\rvax(2),\ldots$.
It is first shown that
Gorbunov \& Pinsker’s results on the stationarity of the realizations to the causal IRDF (stated for two-sided stationary sources) do not apply to the commonly used family of 
\textit{asymptotic average single-letter} (AASL) distortion criteria.
Moreover, we show that, in general, a reconstruction sequence cannot be both jointly stationary with a one-sided stationary source sequence and causally related to it.
This implies that, in general, the causal IRDF for one-sided stationary sources cannot be realized by a stationary distribution.
However, we prove that for an arbitrarily distributed one-sided stationary source and a large class of distortion criteria (including AASL), the search for $R_{c}^{it}(D)$ can be restricted to distributions which yield the output sequence $\rvay_{1}^{\infty}$ jointly stationary with the source  after $\kappa$ samples. 
Finally, we improve the definition of 
the stationary causal IRDF $\overline{R}_{c}^{it}(D)$ previously introduced 
by Derpich and {\O}stergaard for two-sided Markovian stationary sources and show that
$\overline{R}_{c}^{it}(D)$ for a two-sided source $\ldots,\rvax(-1),\rvax(0),\rvax(1),\ldots$
equals
$R_{c}^{it}(D)$ for the associated one-sided source $\rvax(1),\rvax(2),\ldots$.
This implies that, for the Gaussian quadratic case, 
the practical zero-delay encoder-decoder pairs proposed by Derpich \& {\O}stergaard for approaching $\overline{R}_{c}^{it}(D)$ achieve an operational data rate which exceeds $R_{c}^{it}(D)$ by less than
$1+0.5\log_{2}(2\pi\expo{}/12)\simeq 1.254$ bits per sample.
\end{abstract}

\section{Introduction}\label{sec:Intro}
The \textit{information} RDF (IRDF) $R^{it}(D)$ 
for a given one-sided
random source process $\rvax_{1}^{\infty}\eq\set{\rvax(1),\rvax(2),\ldots}$ can be defined as the infimum of the mutual information rate~\cite[Section~8.2]{gray--11} 
\begin{align}\label{eq:Irate_def}
 \bar{I}(\rvax_{1}^{\infty};\rvay_{1}^{\infty})
 \eq 
 \lim_{k\to\infty}\frac{1}{k} I(\rvax_{1}^{k};\rvay_{1}^{k})
\end{align}
between source and reconstruction $\rvay_{1}^{\infty}$ such that a given 
fidelity criterion does not exceed a distortion value $D$~\cite{berger71,covtho06,gray--11}.
If one adds to this definition the restriction that the decoder output can only depend causally upon the source, one obtains what is known as the causal~\cite{neugil82,derost12}, non-anticipative~\cite{gorpin73,pingor87,chasta14} or sequential IRDF~\cite{tatiko03,tatsah04,tankim15}.
All these are equivalent and will be denoted as $R_{c}^{it}(D)$, defined 
in terms of the mutual information $I(\rvax_{1}^{n};\rvay_{1}^{n})$~\cite{gray--11,covtho06}
as
\begin{align}\label{eq:Causal_IRDF_base}
 R_{c}^{it}(D)\eq \inf  \bar{I}(\rvax_{1}^{\infty};\rvay_{1}^{\infty}),
\end{align}
where the infimum is taken over all joint distributions of $\rvay_{1}^{\infty}$ given $\rvax_{1}^{\infty}$ such that the causality Markov chains 
(which will be referred to as the \textit{short} causality constraint)
\begin{align}\label{eq:MC_causality_0}
\rvax_{k+1}^{\infty} \longleftrightarrow 
\rvax_{1}^{k} \longleftrightarrow    \rvay_{1}^{k}, \fspace k=1,2,\ldots.
\end{align}
hold and which yield distortion 
not greater than $D$, for some fidelity criterion.
%
Notice that, if one is given a two-sided random source process 
$\rvax_{-\infty}^{\infty}=\set{\ldots,\rvax(-1),\rvax(0),\rvax(1),\ldots}$ instead, and one is interested only in encoding and reconstructing the samples $\rvax_{1}^{\infty}$, then the causality constraints may be stated as 
\begin{align}\label{eq:MC_causality_2sdd}
\rvax_{k+1}^{\infty} \longleftrightarrow 
\rvax_{-\infty}^{k} \longleftrightarrow    \rvay_{1}^{k}, \fspace k=1,2,\ldots.
\end{align}
as done in~\cite{gorpin73,pingor87,derost12}.
This notion of causality will be referred to as the \textit{long} causality constraint.

The motivation for considering in this work one-sided instead of two-sided sequences (and thus~\eqref{eq:MC_causality_0} instead of~\eqref{eq:MC_causality_2sdd}) arises from the aim of building encoder-decoder systems which operate with zero delay (the same motivation behind the causality constraint).
 To see this, notice that the causality constraint~\eqref{eq:MC_causality_2sdd} for two-sided sources 
 corresponds to the situation in which source samples in the infinite past exist and are available to the encoder.
 This may require an infinite delay before actually beginning to encode and decode.
 By contrast, the causality constraint~\eqref{eq:MC_causality_0} describes the case when the source is a one-sided process and $\rvay_{1}^{k}$ depends only upon $\rvax_{1}^{k}$ (as in~\cite{chasta14,woolin17}).

\begin{rem}\label{rem:strong_causality}
 It is important to highlight at this point that even though the causality condition~\eqref{eq:MC_causality_0} can also be applied to a two-sided source process $\rvax_{-\infty}^{\infty}$, it would not ensure causality in that case.
To see why, consider the situation in which $\rvax_{-\infty}^{\infty}$ is a binary i.i.d. source where each $\rvax_{k}$ takes the values $1$ or $0$ with equal probability.
Suppose $\rvay_{1}^{\infty}$ is built as $\rvay(k) =\rvax_{-k}\oplus \rvax_{2k}$, where $\oplus$ denotes the exclusive ``OR'' operator.
It is easy to see that $(\rvax_{1}^{\infty},\rvay_{1}^{\infty})$ satisfies~\eqref{eq:MC_causality_0}, even though $\rvay_{1}^{\infty}$ depends non-causally on 
$\rvax_{1}^{\infty}$.

The above observation reveals that if the source is two sided but 
only the samples $\rvax_{1}^{\infty}$ are encoded and the decoded process is one-sided ($\rvay_{1}^{\infty}$), then 
one needs to impose instead the (more general) causality constraint%
\begin{align}\label{eq:MC_causality_hybrid}
(\rvax_{-\infty}^{0},
\rvax_{k+1}^{\infty})
\longleftrightarrow 
\rvax_{1}^{k} \longleftrightarrow    \rvay_{1}^{k}, \fspace k=1,2,\ldots.
\end{align}
which implies~\eqref{eq:MC_causality_0}.
Besides causality, these Markov chains guarantee that even if the source is a two-sided process, its encoding and reconstruction proceeds as if it were a one-sided process.

Notice that~\eqref{eq:MC_causality_hybrid} implies~\eqref{eq:MC_causality_0} and~\eqref{eq:MC_causality_2sdd}.
For this reason,~\eqref{eq:MC_causality_hybrid} will be referred to as the \textbf{strong} causality constraints.
\end{rem}

As we shall see in sections~\ref{sec:joint_stat_and_causality} and~\ref{sec:quasi_stat_is_sufficient}, this situation, where at time $k$ the encoder can take only $\rvax_{1}^{k}$ as input, entails significant challenges due to the unavoidable need to deal with transient phenomena.

The operational significance of $R_{c}^{it}(D)$ stems from its relation to the 
causal \textit{operational} RDF (ORDF), denoted as  
$R^{op}_{c}(D)$.
The latter is defined as the infimum of the average data-rates which are achievable by a sequence of causal encoder-decoder functions~\cite{neugil82,derost12} yielding a distortion not greater than $D$.
Characterizing $R^{op}_{c}(D)$ is important because every zero-delay source code 
(suitable for applications such as low-delay streaming~\cite{chenwu15} or networked control~\cite{naifag07,silder16}) must be causal.

An IRDF is said to be achievable if it equals the ORDF under the same constraints~\cite{berger71,covtho06}.
As far as the authors are aware, the achievability of $R_{c}^{it}(D)$ has not been demonstrated yet, 
for any source and distortion measure, and thus the gap between $R_{c}^{op}(D)$ and $R_{c}^{it}(D)$ is unknown in general.
However, it is known that~\cite[Section~II]{derost12} 
\begin{align}
  R_{c}^{op}(D)\geq R_{c}^{it}(D)
\end{align}
and for Gaussian sources it is possible to construct causal codes with an operational data rate exceeding  
$R_{c}^{it}(D)$ by less than (approximately) 0.254 bits/sample (1.254 bits/sample for zero-delay codes), once the statistics which realize the latter are known~\cite{derost12}.
This underlines the importance of studying the causal IRDF $R_{c}^{it}(D)$.

To the best of the authors' knowledge, no closed-form expressions are known for $R_{c}^{it}(D)$, except when considering \textit{mean-squared-error} (MSE) distortion and for Gaussian i.i.d. or Gaussian \textit{auto-regressive} (AR)-1 sources, either scalar~\cite[Section~IV]{derost12} or vector valued~\cite{staost17}%
\footnote{Although for i.i.d. sources and for a single-letter distortion criterion a realization of the (non-causal) RDF satisfies causality~\cite{covtho06,berger71}, the formulas available in the literature for expressing it require numerical iterative procedures and cannot be regarded as ``closed-form'' except for the Gaussian case and MSE distortion.}.
However, there exist various structural properties of the causal IRDF that have been found in literature when $R_{c}^{it}(D)$ admits (or is assumed to admit) a stationary realization.

Indeed, the stationarity of the realizations of the causal IRDF has played a crucial role in simplifying the computation of $R_{c}^{it}(D)$ for Gaussian 1-st order Markovian sources and MSE distortion in~\cite{tanaka15}.
 It has also been a key implicit assumption in~\cite{tatsah04}, and an explicit assumption in works such as~\cite{chasta14} and~\cite{derost12}.
 In particular, for a stationary two-sided random source $\rvax_{-\infty}^{\infty}$,~\cite[Definition~6]{derost12} introduced the stationary causal IRDF 
 \begin{align}\label{eq:overlineRcitD_def}
  \overline{R}_{c}^{it}(D)\eq \inf \bar{I}(\rvax_{1}^{\infty};\rvay_{1}^{\infty})
 \end{align}
where the infimum is taken over all 
distributions of $\rvay_{1}^{\infty}$ given $\rvax_{1}^{\infty}$ which yield a one-sided reconstruction processes $\rvay_{1}^{\infty}$ jointly stationary with $\rvax_{1}^{\infty}$, satisfying~\eqref{eq:MC_causality_2sdd} and an asymptotic average MSE distortion constraint on $(\rvax_{1}^{\infty},\rvay_{1}^{\infty})$.
For the case of a Gaussian source, it was shown in~\cite{derost12} that an operational data-rate exceeding $\overline{R}_{c}^{it}(D)$ by less than $1+0.5\log_{2}(2\pi\expo{})\simeq 1.254$ bits/sample was achievable using a \textit{entropy-coded subtractively dithered uniform quantizer} 
(ECSDUQ)
surrounded by 
\textit{linear time-invariant} (LTI)
filters operating in steady state.
These examples illustrate the relevance of determining whether (or in which cases) the causal IRDF admits a stationary realization.
 
To the best of our knowledge, the only work which has given an answer to this question in a general framework is~\cite{gorpin73}.
Under a set of assumptions (discussed in Section~\ref{sec:revisiting_gorpin73} below), it is shown in~\cite[Theorem 4]{gorpin73} that the search for the causal IRDF for a large class of two-sided sources and distortion criteria can be restricted to reconstructions which are jointly stationary with the source.
Unfortunately, as we show in Section~\ref{subsec:analysis_of_gorpin73}, the assumptions on the fidelity criteria utilized in~\cite{gorpin73} leave out some common distortions (such as the family of asymptotic average single-letter fidelity criteria), and the statement of~\cite[Theorem 4]{gorpin73} contains an assumption whose validity has to be proved.
More importantly, the entire analysis of~\cite{gorpin73} is built for two-sided processes (using the causality constraint~\eqref{eq:MC_causality_2sdd}), which opens the question of whether its results could apply to  one-sided processes as well, with the causality constraint~\eqref{eq:MC_causality_0}.

In this paper we give an answer to these questions and use the results to prove some novel properties of the causal IRDF associated with the stationarity of its realizations.
Specifically, our main contributions are the following:

\begin{enumerate}
 \item We show in Theorem~\ref{thm:conditions_forJS_and_causal} that if a  pair of one-sided random processes $\rvax_{1}^{\infty},\rvay_{1}^{\infty}$ is jointly stationary, with the latter depending causally on the former according to~\eqref{eq:MC_causality_hybrid} (but otherwise arbitrarily distributed), then it must also satisfy the Markov chains
 \begin{align}
   \rvax_{k+1}^{\infty} 
   \longleftrightarrow
   \rvax(k)
   \longleftrightarrow
   \rvay(k),
   \fspace 
    k=1,2,\ldots,
 \end{align}
which is a fairly restrictive condition.
In particular, as we show in Theorem~\ref{thm:stationarity_and_causality_Gaussian}, if $\rvax_{1}^{\infty},\rvay_{1}^{\infty}$ are jointly Gaussian and  $\rvay_{1}^{\infty}$ depends causally upon $\rvax_{1}^{\infty}$, then joint stationarity implies $\rvax_{1}^{\infty}$ is an i.i.d. or 1st-order Markovian process.
This stands in stark contrast with what was shown in~\cite{gorpin73} for two-sided stationary processes and constitutes a counterexample of what is stated in~\cite[Theorem~III.6]{stakou15}.

\item 
Despite the above, we show in Theorem~\ref{thm:QJS_is_sufficient}
that for any $\kappa$-th order Markovian one-sided stationary source $\rvax_{1}^{\infty}$ and a large class of distortion constraints,
the search for the causal IRDF (as defined in~\eqref{eq:Causal_IRDF_base}) can be restricted to output sequences causally related to the source and 
jointly stationary with it after $\kappa$ samples, and such that 
$
\bar{I}(\rvax_{\kappa}^{\infty};\rvay_{\kappa}^{\infty})
=
\bar{I}(\rvax_{1}^{\infty};\rvay_{1}^{\infty})$.
We refer to such pairs of processes as being $\kappa$-\textit{quasi-jointly stationary} ($\kappa$-QJS) (this notion is formally introduced in Definition~\ref{def:quasi_jointly_stat} below).
A consequence of this result is that 
for any $\kappa$-th order two-sided Markovian stationary source $\rvax_{-\infty}^{\infty}$,
$\overline{R}_{c}^{it}(D)$ 
equals $R_{c}^{it}(D)$ for the corresponding one-sided stationary source $\rvax_{1}^{\infty}$.
The relevance of this finding is that for Gaussian stationary sources and asymptotic MSE distortion, an operational data rate exceeding 
$\bar{I}(\rvax_{1}^{\infty},\rvay_{1}^{\infty})$ (and thus $R_{c}^{it}(D)$)  
by less than approximately 
$0.254$ bits/sample, when operating causally, and 
$1.254$ bit/sample, in zero-delay operation, is achievable  
by using a scalar ECSDUQ as in~\cite{derost12}.
\end{enumerate}

The remainder of this paper begins with Section~\ref{sec:revisiting_gorpin73}, in which the assumptions leading to~\cite[Theorem 4]{gorpin73} are revisited and the limitations of that theorem are discussed.
In Section~\ref{sec:joint_stat_and_causality} we prove that, in general, it is not possible to have two one-sided processes which are jointly stationary and, at the same time, satisfy the causality constraint~\eqref{eq:MC_causality_0}.
Section~\ref{sec:quasi_stat_is_sufficient} presents our main theorem (Theorem~\ref{thm:QJS_is_sufficient}), which shows that the search for the causal IRDF for one-sided $\kappa$-th order Markovian stationary sources can be restricted to $\kappa$-QJS processes.
Finally, Section~\ref{sec:Conclusions} draws the main conclusions of this work.
All proofs are presented in section~\ref{sec:appendix} (the Appendix), which also contains some technical lemmas required by these proofs.

\paragraph*{Notation}
$\Rl$ denotes the real numbers, $\Z$ denotes the integers, $\Nl=\Z^{+}$ is the set of natural numbers (positive integers), 
$\Z_{0}^{-}\eq \set{\ldots,-2,-1,-0}$ and $\Nl_{0}\eq\set{0,1,\ldots}$.
For every $x\in\Rl$, the ceiling operator $\lceil x\rceil$ yields the smallest integer not less than $x$.  
We use non-italic letters for scalar random variables, such as $\rvax$.
Random sequences are denoted as $\rvax_{1}^{k}=\set{\rvax(i)}_{i=1}^{k}=\set{\rvax(1),\rvax(2),\ldots,\rvax(k)}$.
For a random (one-sided) process $\rvax_{1}^{\infty}$
we will sometimes use the short-hand notation $\rvax$ wherever this meaning is clear from the context.
When convenient, we write a random sequence $\rvay_{j}^{k}$, $j\leq k$, as the column vector
$\rvey^{j}_{k}\eq [\rvay(j)\; \rvay(j+1)\;\cdots\;\rvay(k)]^{T}$ (the indices $j$ and $k$ are swapped so that the smallest index goes above the largest one, thus mimicking the usual index order in a column vector). 
The entry on the $f$-th row and $k$-th column of a matrix $\bM$ is denoted as $[\bM]_{f,k}$, with 
$[\bM]^{j}_{k}$ being the sub-matrix of $\bM$ containing its rows $j$ to $k$, $j\leq k$.  

For a random element $\rvax$ in a given alphabet (set) $\Xsp$, we write $\Bsp(\Xsp)$ to denote a sigma-algebra associated with $\Xsp$ and $P_{\rvax}:\Bsp(\Xsp)\to[0,1]$ to denote its probability distribution (or probability measure).
We write $\rvax \sim \rvay$ to describe the fact that $\rvay$ has the same probability distribution as $\rvax$,
and $\rvax\Perp \rvay$ to state that $\rvax$ and $\rvay$ are independent.
We write the condition in which two random elements $\rvaa,\rvab$ are independent given a third random element $\rvac$  using the Markov chain notation 
$
\rvaa
\longleftrightarrow 
\rvac
\longleftrightarrow 
\rvab$.
If $\Wsp$ is a set of probability distributions, then $(\Wsp)$ denotes the set of all random elements whose probability distribution belongs to $\Wsp$.
The expectation operator is denoted as $\Expe{\cdot}$.
We write $\Xsp^{k}$ as a shorthand for $\times_{i=1}^{k}\Xsp$.
The mutual information between two random elements 
$\rvax\in \Xsp$ $\rvay\in \Ysp$
is defined as%
~\cite[Lemma 7.14]{gray--11}
\begin{align}
 I(\rvax;\rvay)
 \eq 
 \sup_{q,r}\Expe{
 \log\left(
 \frac{P_{q(\rvax),r(\rvay)}(q(\rvax),r(\rvay))} 
      {P_{q(\rvax)}(q(\rvax)) P_{r(\rvay)}(r(\rvay))}\right)  
 },
\end{align}
where the supremum is over all quantizers $q$ and $r$ of $\Xsp$ and $\Ysp$, and  
$P_{q(\rvax),r(\rvay)}$, $P_{q(\rvax)}$ and $P_{r(\rvay)}$, are the joint and marginal distributions of  
$q(\rvax)$ and $r(\rvay)$, respectively.
If $\rvax,\rvay$ have 
joint and marginal \textit{probability density functions} (PDFs)
$f_{\rvax,\rvay}$,
$f_{\rvax}$ and 
$f_{\rvay}$, respectively, then~\cite{covtho06} 
$$
I(\rvax;\rvay)\eq \Expe{\log\left(\frac{f_{\rvax,\rvay}(\rvax,\rvay)}{f_{\rvax}(\rvax) f_{\rvay}(\rvay)}\right)}.
$$
The conditional mutual information 
$I(\rvaa;\rvab|\rvac)$ 
is defined via the chain-rule (cr) of mutual information 
$I(\rvaa;\rvab|\rvac)\eq I(\rvaa;\rvab,\rvac) - I(\rvaa;\rvac)$.
The mutual information rate $\bar{I}(\rvax_{1}^{\infty};\rvay_{1}^{\infty})$ between two processes $\rvax_{1}^{\infty}$ and $\rvay_{1}^{\infty}$ is defined as in~\eqref{eq:Irate_def}.
The variance of a real-valued random variable $\rvax$ is denoted as 
$\sigsq_{\rvax}=\Expe{(\rvax-\expe{\rvax})^{2} }$.
The auto-correlation function of a random process $\rvax_{1}^{\infty}$ is denoted 
$\varrho_{\rvax}(\tau,k)\eq \Expe{\rvax(k)\rvax(k+\tau)}$, $k\geq 1$, $\tau>-k$.

The following properties of the mutual information involving any random elements $\rvaa,\rvab,\rvac$  will be utilized and referred to throughout this work:
\begin{property}\label{property:Iabc_larger_than_Iab}
$I(\rvaa;\rvab,\rvac)\geq I(\rvaa;\rvab)$, with equality if and only if $I(\rvaa;\rvac|\rvab)=0$.
\end{property}
\begin{property}\label{property:Iabc_larger_than_Iab_gvn_c}
$I(\rvaa;\rvab,\rvac)\geq I(\rvaa;\rvab|c)$, with equality if and only if $I(\rvaa;\rvac)=0$.  
\end{property}
We will also make use of the following fact:
\begin{fact}\label{fact:abc_and_barabc}
Let $\rvaa,\rvab,\rvac$ be three random elements with an arbitrary joint distribution. 
Then, there exists a random element $\bar{\rvaa}$ 
(equivalently, a joint distribution $P_{\bar{\rvaa},\rvab,\rvac}$) 
such that 
\begin{align}
  (\bar{\rvaa},\rvab)\sim  (\rvaa,\rvab)\\
  \bar{\rvaa} \longleftrightarrow \rvab \longleftrightarrow \rvac
\end{align}
\end{fact}

\section[Revisiting Gorbunov and Pinsker 1973 Paper]{Revisiting~\cite{gorpin73} and its Inapplicability to One-Sided Sources}\label{sec:revisiting_gorpin73}
In order to assess whether (or to what extent)~\cite[Theorem 4]{gorpin73} could provide support to the stationarity assumptions made in, e.g.~\cite{tatsah04,dersil08,chasta14,stakou15}, it is necessary to take a closer look at the assumptions made in~\cite{gorpin73}
and the statement of its Theorem 4. 
For that purpose, the first part of this section is an exposition of the definitions and assumptions leading to~\cite[Theorem 4]{gorpin73}.%
\footnote{We believe this re-exposition of~\cite{gorpin73} to be valuable in itself since on the one hand, it selects the minimal set of notions required to formulate and understand its Theorem 4, and on the other hand, it provides an arguably clearer presentation than the one found in~\cite{gorpin73} (an English translation from Russian), which is not easy to read due to its notation, some mathematical typos and the low resolution of its available digitized form.}
The second part is an analysis which reveals the limitations of~\cite[Theorem 4]{gorpin73} and its inapplicability to the case in which the source and reconstruction are one-sided processes.
At the same time, this section also introduces definitions and part of the notation to be utilized in the remainder of this paper (for convenience, a summary of these is presented in Table~\ref{table:symbols} below).

\subsection{A Brief Review of~\cite{gorpin73}}
Throughout~\cite{gorpin73}, the search in the infimizations associated with various types of ``nonanticipatory'' (i.e., causal) rate-distortion functions is stated over sets of \textit{joint} probability distributions between source and reconstruction (as opposed to the usual definitions, in which the search is over \textit{conditional} distributions, see~\eqref{eq:Causal_IRDF_base} and~\cite[Chapter~10]{covtho06},~\cite{berger71}). 
Since the distribution of the source is given, it is required that for every $k_{2}>k_{1}\in\Z$, all the joint distributions $P_{\rvax_{k_{1}}^{k_{2}},\rvay_{k_{1}}^{k_{2}}}$
to be considered yield $\rvax_{k_{1}}^{k_{2}}$ 
having the same (given) distribution of the source for the corresponding block, say $P_{\mathring{\rvax}_{k_{1}}^{k_{2}}}$.
This requirement can be formalized 
as requiring that $P_{\rvax_{k_{1}}^{k_{2}},\rvay_{k_{1}}^{k_{2}}}\in \Pspa^{k_{1},k_{2}}$, 
for a set of admissible joint distributions $\Pspa^{k_{1},k_{2}}$ defined as
\begin{align}\label{eq:Pkk_def}
 \Pspa^{k_{1},k_{2}}\eq \Set{P: P(E\times \Ysp_{k_{1}}^{k_{2}})  = P_{\mathring{\rvax}_{k_{1}}^{k_{2}}}(E),\quad \forall E \in\Bsp(\Xsp_{k_{1}}^{k_{2}}) },\fspace k_{1}\leq k_{2}\in\Z,
\end{align}
where $\Xsp_{k_{1}}^{k_{2}}$ and $\Ysp_{k_{1}}^{k_{2}}$ are, respectively, the alphabets to which $\rvax_{k_{1}}^{k_{2}}$  and
$\rvay_{k_{1}}^{k_{2}}$
belong.
In~\cite{gorpin73}, this admissibility requirement is embedded in the definition of the sets of distributions which meet the distortion constraint, described next.

The fidelity criterion for every pair 
of integers%
\footnote{The analysis in~\cite{gorpin73} considered both discrete- and continuous-time processes, but here we only refer to the discrete-time scenario.} 
$k_{1}\leq k_{2}$ 
is expressed in~\cite{gorpin73} as requiring $P_{\rvax_{k_{1}}^{k_{2}},\rvay_{k_{1}}^{k_{2}}}$ to belong to a non-empty set of distributions (hereafter referred to as \textit{distortion-feasible set}) $\Wsp^{k_{1},k_{2}}_{D}$,
a condition written as 
$(\rvax_{k_{1}}^{k_{2}},\rvay_{k_{1}}^{k_{2}})\in (\Wsp^{k_{1},k_{2}}_{D})$.
In this definition, 
the number $D\geq 0$ represents an admissible distortion level. 
Notice that such general formulation of a fidelity criteria does not need a distortion function and does not necessarily involve an expectation.

As mentioned above, the admissibility  requirement $P_{\rvax_{k_{1}}^{k_{2}}}\in\Pspa^{k_{1},k_{2}}$ is expressed in the distortion-feasible sets in~\cite[eqn. (2.1)]{gorpin73}.
The latter equation can be written as 
\begin{align}\label{eq:W_D_k_k_right_marginal_cond}
\Wsp_{D}^{k_{1},k_{2}}\subset \Pspa^{k_{1},k_{2}}.
\end{align}

In~\cite[eqs. (2.4) and (2.5)]{gorpin73}, the distortion-feasible sets are assumed to satisfy the ``concatenation'' condition
\begin{align}
\label{eq:concatenation}
 (\rvax_{k_{1}}^{k_{2}},\rvay_{k_{1}}^{k_{2}})\in(\Wsp_{D}^{k_{1},k_{2}})
 \wedge
 (\rvax_{k_{2}+1}^{k_{3}},\rvay_{k_{2}+1}^{k_{3}})\in(\Wsp_{D}^{k_{2}+1,k_{3}})
 \Longrightarrow
 (\rvax_{k_{1}}^{k_{3}},\rvay_{k_{1}}^{k_{3}})\in(\Wsp_{D}^{k_{1},k_{3}}).
\end{align}

With this,~\cite[eqn.~(2.9)]{gorpin73} defined the 
``nonanticipatory epsilon entropy'' of the set of distributions%
\footnote{The actual term employed in~\cite{gorpin73} is ``nonanticipatory epsilon entropy of the message $(\Wsp_{D}^{k_{1},k_{2}})$'' where the term ``message'' refers to the random ensembles in $(\Wsp_{D}^{k_{1},k_{2}})$.}
$\Wsp_{D}^{k_{1},k_{2}}$ as
\begin{align}
H^{0}(\Wsp_{D}^{k_{1},k_{2}} )\eq  \inf I(\rvax_{k_{1}}^{k_{2}};\rvay_{k_{1}}^{k_{2}}),
\end{align}
where the infimum is taken over all pairs of random sequences 
$(\rvax_{k_{1}}^{k_{2}},\rvay_{k_{1}}^{k_{2}})\in(\Wsp_{D}^{k_{1},k_{2}})$ such that the causality Markov chains
\begin{align}
 \rvax_{k+1}^{k_{2}}
 \longleftrightarrow \rvax_{k_{1}}^{k} \longleftrightarrow 
 \rvay_{k_{1}}^{k}
 , \fspace k_{1}\leq k\leq k_{2}
\end{align}
are satisfied.
Then~\cite[eq.~(2.13)]{gorpin73} defines the ``nonanticipatory message generation rate'' 
as
\begin{align}\label{eq:HD0_def}
 \overline{H^{0}_{D}}\eq \lim_{k_{2}-k_{1}\to\infty} \frac{1}{k_{2}-k_{1}}H^{0}(\Wsp_{D}^{k_{1},k_{2}} )
\end{align}
(when the limit exists).
An alternative ``nonanticipatory message generation rate'' is also considered in~\cite{gorpin73} by defining the set of distortion-admissible process distributions $\Wsp_{D}$ as follows:

\begin{defn}\label{def:W_D}
 The set $(\Wsp_{D})$ consists of all two-sided random process pairs 
 $(\rvax_{-\infty}^{\infty},\rvay_{-\infty}^{\infty})\in(\Wsp_{D})$ for which there exist integers $\cdots< k_{-1}<k_{0}<k_{1}<\cdots$ such that
$\lim_{i\to\pm \infty}k_{i}=\pm\infty$ and 
\begin{align}\label{eq:WD_redef}
 (\rvax_{k_{i}}^{k_{i+1}-1},\rvay_{k_{i}}^{k_{i+1}-1})\in (\Wsp_{D}^{k_{i},k_{i+1}-1}),\fspace \forall i\in\Z.\\[-12mm]\nonumber
\end{align}
\finenunciado
\end{defn}

With this,~\cite[eq.~(2.12)]{gorpin73} defines 
\begin{align}\label{eq:Harrow_def}
 \overrightarrow{H_{D}^{0}}\eq \inf \lim_{k_{2}-k_{1}\to\infty} 
 \frac{1}{k_{2}-k_{1}}
 I(\rvax_{k_{1}}^{k_{2}};\rvay_{k_{1}}^{k_{2}})
\end{align}
(when the limit exists), where the infimum is taken over all pairs of processes 
$(\rvax_{-\infty}^{\infty},\rvay_{-\infty}^{\infty})\in(\Wsp_{D})$ satisfying the causality Markov chains 
\begin{align}\label{eq:MC_causality_gorpin}
\rvax_{k+1}^{\infty}
\longleftrightarrow
\rvax_{-\infty}^{k}
\longleftrightarrow
\rvay_{-\infty}^{k}
,\fspace 
\forall k\in\Z.
\end{align}
Notice that these Markov chains imply~\eqref{eq:MC_causality_2sdd} and differ from the latter in that here the reconstruction $\rvay_{-\infty}^{\infty}$ is a two-sided random process.

Now assume that $\Xsp_{i}=\Xsp$ and $\Ysp_{i}=\Ysp$, for all $i\in\Z$, for some alphabets $\Xsp$ and $\Ysp$.
Define, for any given non-negative sequence $\set{a(s)}_{s=s_{1}}^{s_{2}}$, $s_{1}\leq s_{2}$ such that $1=\sumfromto{s=s_{1}}{s_{2}}a(s)$, the distribution
\begin{align}
 P_{\bar{\rvax}_{k_{1}}^{k_{2}},\bar{\rvay}_{k_{1}}^{k_{2}}}(E)
= 
 \Sumfromto{s=s_{1}}{s_{2}}
 a(s)
 P_{\rvax_{k_{1}+s}^{k_{2}+s},\rvay_{k_{1}+s}^{k_{2}+s}}(E)\label{eq:P_bar}, \fspace k_{1}\leq k_{2}, E\in\Bsp(\Xsp^{(k_{2}-k_{1}+1)}\times \Ysp^{(k_{2}-k_{1}+1)}).
\end{align}
We can now 
re-state Theorem 4 in~\cite{gorpin73} as follows:
\begin{thm}[Theorem~4~in\cite{gorpin73}]\label{thm:4_gorpin}
 Suppose that 
 \begin{enumerate}[i)]
  \item $\rvax_{-\infty}^{\infty}$ is stationary.
  \item Stationary distortion-feasible sets: For every $s\in\Z$, $k_{1} \leq k_{2}\in\Z$, 
  $\Wsp_{D}^{k_{1}+s,k_{2}+s}$ and   $\Wsp_{D}^{k_{1},k_{2}}$ are identical sets.%
  \footnote{In~\cite{gorpin73} this condition together with the stationarity of $\rvax_{-\infty}^{\infty}$ is referred to as ``a stationary source'' (see its description between (2.8) and (2.9) in~\cite{gorpin73}).}
 \item The concatenation condition~\eqref{eq:concatenation} holds.
  \item $\overrightarrow{H^{0}_{D}} = \overline{H_{D}^{0}}$.
  \item For every set of non-negative numbers
  $\set{a(s)}_{s=s_{1}}^{s_{2}}$, $s_{1}\leq s_{2}$ such that $1=\sumfromto{s=s_{1}}{s_{2}}a(s)$, 
  \begin{align}
    (\rvax_{-\infty}^{\infty},\rvay_{-\infty}^{\infty})\in(\Wsp_{D})
 \Longrightarrow
 (\bar{\rvax}_{-\infty}^{\infty},\bar{\rvay}_{-\infty}^{\infty})\in(\Wsp_{D})
  \end{align}
where the processes $(\bar{\rvax}_{-\infty}^{\infty},\bar{\rvay}_{-\infty}^{\infty})$ are distributed according to~\eqref{eq:P_bar}.  
 \end{enumerate}
Then, the analysis of the lower bound in~\eqref{eq:Harrow_def} can be confined to jointly stationary pairs of random processes 
$(\rvax_{-\infty}^{\infty},\rvay_{-\infty}^{\infty})\in(\Wsp_{D})$ satisfying the causality constraint~\eqref{eq:MC_causality_gorpin}. 
\finenunciado
\end{thm}

For convenience, Table~\ref{table:symbols} presents a summary of the definitions and notation described so far, together with some
which will be defined in the following sections.

\begin{table}[htbp]
\caption{Summary of the main symbols utilized in this paper.}
\label{table:symbols}
\centering
\begin{tabular}{|p{3cm}|p{12cm}|}
\hline 
%
$P_{\rvax_{k_1}^{k_2},\rvay_{k_1}^{k_2}}$ 
& 
The joint probability distribution of $\rvax_{k_1}^{k_2},\rvay_{k_1}^{k_2}$ \\ 
\hline 
%
$\Pspa^{k_1,k_2}$ 
& 
The set of all joint distributions $P_{\tilde{\rvax}_{k_1}^{k_2},\tilde{\rvay}_{k_1}^{k_2}}$ such that the associated marginal distribution
$P_{\tilde{\rvax}_{k_{1}}^{k_{2}}}$ equals the given distribution of the source sequence 
$\rvax_{k_{1}}^{k_{2}}$, i.e., $P_{\tilde{\rvax}_{k_{1}}^{k_{2}}}=P_{\rvax_{k_{1}}^{k_{2}}}$ (see~\eqref{eq:Pkk_def}). 
\\ 
\hline 
%
$\Wsp_D^{k_1,k_2}$
& 
Distortion-feasible set.
The set of all joint distributions $P_{\rvax_{k_1}^{k_2},\rvay_{k_1}^{k_2}}\in\Pspa^{k_1,k_2}$ which satisfy a given constraint given by $D\in\Rl$ (see comments before~\eqref{eq:Pkk_def}).
\\ 
\hline 
%
$(\Wsp_D^{k_1,k_2})$ 
& 
The set of all pairs of sequences 
$(\tilde{\rvax}_{k_{1}}^{k_{2}},\tilde{\rvay}_{k_{1}}^{k_{2}})$
such that 
$P_{\tilde{\rvax}_{k_{1}}^{k_{2}},\tilde{\rvay}_{k_{1}}^{k_{2}}}\in \Wsp_D^{k_1,k_2}$.
(See also the Notation subsection at the end of Section~\ref{sec:Intro}.)
\\ 
\hline 
%
$\Pspa_D^{\infty}$  
&
Generic distortion-feasible set of probability distributions for pairs of one-sided processes
$(\rvax_1^{\infty},\rvay_1^{\infty})$.
In this paper, we state some minimal conditions on $\Pspa_D^{\infty}$ in Assumption~\ref{assu:Pspa_D_inf_assu} and some additional structural properties in Assumption~\ref{assu:DFS_as_ineq_for_k}.
\\ 
\hline 
%
$\Qsp_{\kappa}$, $\kappa=1,2,\ldots$
&
The set of all joint distributions
$P_{\rvax_{1}^{\infty},\rvay_{1}^{\infty}}$ of pairs of one-sided random processes 
$(\rvax_{1}^{\infty},\rvay_{1}^{\infty})$
such that
$
(\rvax_{\kappa}^{\infty},\rvay_{\kappa}^{\infty})
$
are jointly stationary and
$
\lim_{\ell\to\infty}
\frac{1}{\ell}I(\rvax_{1}^{\ell};\rvay_{1}^{\ell})
=
\lim_{\ell\to\infty}
\frac{1}{\ell}I(\rvax_{\kappa}^{\kappa+\ell-1};\rvay_{\kappa}^{\kappa+\ell-1})
$
(see Definition~\ref{def:quasi_jointly_stat}).
\\ 
\hline 
%
$(\Csp^n),$ $n = 1,2,\ldots$ and
&
The sets of causally related one-sided pairs of $n$-sequences (see Definition~\ref{def:Csp}).
\\ 
\hline 
%
$(\Csp^{\infty})$
&
The set of one-sided pairs of processes 
causally related 
according to the \textbf{short causality constraint}~\eqref{eq:MC_causality_0} (see Definition~\ref{def:Csp}).
\\
\hline 
%
%
$\Csp_{-\infty}^{\infty}$
&
The set of causal distributions for processes of the form 
$(\rvax_{-\infty}^{\infty},\rvay_1^{\infty})$.
Such processes satisfy the \textbf{long causality constraint}~\eqref{eq:MC_causality_2sdd} 
(see Definition~\ref{def:Overline_RCitd_redef}).
\\ 
\hline  
\end{tabular} 
\end{table}

\subsection{Analysis of Theorem~\ref{thm:4_gorpin} and its Inapplicability to One-Sided Sources}\label{subsec:analysis_of_gorpin73}
We now discuss three limitations of Theorem~\ref{thm:4_gorpin} which are relevant when trying to establish whether the causal IRDF of a one-sided stationary source admits a stationary realization. 

\paragraph*{Limitation 1}
The first obvious limitation is that even if source and reconstruction are two-sided processes, every distortion criterion which considers only their ``positive-time'' part cannot be expressed by a distortion-feasible set $\Wsp_{D}$ given by Definition~\ref{def:W_D} if the sets $\set{\Wsp_{D}^{\ell,j}}_{\ell \leq j\in\Z}$ satisfy condition~ii) in Theorem~\ref{thm:4_gorpin}.
 To see this, notice that if $\ell\leq j<1$, then such distortion criterion (which neglects non-positive times) would require $\Wsp_{D}^{\ell,j}$ to admit all joint probability distributions satisfying~\eqref{eq:W_D_k_k_right_marginal_cond}.
 Combining this with condition~ii) in Theorem~\ref{thm:4_gorpin} yields that every set  
 $\Wsp_{D}^{n,m}=\Wsp_{D}^{\ell,j}$ with $m-n=j-\ell$, which amounts to imposing no restriction on the distortion at all.
 
 It is natural to think that such elemental shortcoming could be avoided by simply replacing condition~ii) in Theorem~\ref{thm:4_gorpin} by a one-sided version of the form: 
 \begin{align}\label{eq:W_k,k_stationary_one_sided}
\text{For every $t_{1}\in\Z$, $s,t_{2}\in\Nl$ such that $t_{1}\leq t_{2}$: 
  $\Wsp_{D}^{t_{1}+s,t_{2}+s}$ and   $\Wsp_{D}^{t_{1},t_{2}}$ are identical sets.}
 \end{align}
 Leaving aside the fact that this alternative condition is not sufficient for Theorem~\ref{thm:4_gorpin} to hold, it is worth pointing out that  using~\eqref{eq:W_k,k_stationary_one_sided},  
 the commonly utilized family of 
 asymptotic single-letter fidelity criteria~\cite{berger71}
 can not be expressed by 
 a distortion-feasible set $\Wsp_{D}$ given by Definition~\ref{def:W_D}, as the following lemma shows (its proof can be found in  Appendix~\ref{proof:of_lem_Wsp_D_cannot_represent_asymptotic_SLDC}). 
 \begin{lem}\label{lem:Wsp_D_cannot_represent_asymptotic_SLDC}
Let $\rho$
be any given distortion functional which takes as argument a joint distribution $P_{\rvax,\rvay}$ and yields a non-negative real value.
Let $(\Asp_{D})$ be the set of all pairs of processes 
$(\rvax_{-\infty}^{\infty},\rvay_{-\infty}^{\infty})$
where $\rvax_{-\infty}^{\infty}$ is stationary,
with pair-wise distributions $\set{P_{\rvax(k),\rvay(k)}}_{k\in\Z}$ 
which satisfy the asymptotic single-letter fidelity criterion 
 \begin{align}\label{eq:asymptotic_distortion_constraint0}
  \lim_{n\to\infty}\frac{1}{n}\sumfromto{k=1}{n}\rho(P_{\rvax(k),\rvay(k)})\leq D.
 \end{align}
Then, there doesn't exist an infinite collection  of distortion-feasible sets $\set{\Wsp_{D}^{k_{1},k_{2}}}_{k_{1}\leq k_{2}\in\Z}$ satisfying~\eqref{eq:W_k,k_stationary_one_sided} such that the associated $\Wsp_{D}$ given by Definition~\ref{def:W_D} satisfies $(\Wsp_{D})=(\Asp_{D})$.
   \finenunciado
 \end{lem}
 
\paragraph*{Limitation 2}
The second limitation associated with Theorem~\ref{thm:4_gorpin} is 
that its application requires one to prove
its condition iv), i.e., the unproven supposition that  $\overrightarrow{H^{0}_{D}} = \overline{H_{D}^{0}}$, holds.
The only work we are aware of which builds upon Theorem~\ref{thm:4_gorpin} is~\cite{stakou15}, and, accordingly,~\cite{stakou15} provides~\cite[Theorem~III.5]{stakou15}, which states that a similar equality holds.
Unfortunately, as shown in~\cite{derpic17}, the proof of~\cite[Theorem~III.5]{stakou15} is flawed.

We note that Lemma~\ref{lem:liminf_equals_inflim} in Section~\ref{sec:suf_cond_for_inflim_eq_liminf}
below 
provides two alternative sufficient conditions for an equality similar to 
$\overrightarrow{H^{0}_{D}} = \overline{H_{D}^{0}}$ (but for one-sided processes) to hold.

\paragraph*{Limitation 3} The third limitation of Theorem~\ref{thm:4_gorpin} 
for its applicability to one-sided sources is the fact that the entire framework built in~\cite{gorpin73} is stated for two-sided processes (and, crucially, for the corresponding causality restriction given by Markov chain~\eqref{eq:MC_causality_gorpin}).
This difference cannot be simply neglected while expecting Theorem~\ref{thm:4_gorpin} to remain valid.
Indeed, as we show in the next section (Theorem~\ref{thm:conditions_forJS_and_causal}), a pair of random processes $(\rvax_{1}^{\infty},\rvay_{1}^{\infty})$ can be jointly stationary and at the same time satisfy the causality Markov chain~\eqref{eq:MC_causality_0} only if $\rvay(k)$ is independent of $\rvax_{k+1}^{\infty}$ when $\rvax(k)$ is given.
Moreover, we prove that joint stationarity and causality are incompatible when the source is a $\kappa$-th order Markovian Gaussian one-sided process with $\kappa>1$.

\section{Conditions for Joint Stationarity and Causality to Hold Together
}\label{sec:joint_stat_and_causality}
In this section we address the question of whether 
there exists a one-sided reconstruction process $\rvay_{1}^{\infty}$
jointly stationary with a source  
$\rvax_{1}^{\infty}$
and which also satisfies the causality constraint~\eqref{eq:MC_causality_0}.

Each source random sample $\rvax(i)$ belongs to some given set (source alphabet) $\Xsp$  and is allowed to have an arbitrary distribution.
Recall that a random process $\rvay_{1}^{\infty}\in\Ysp^{\infty}$, 
where $\Ysp$ is the reconstruction alphabet 
and
$\Ysp^{\infty}\eq \Ysp\times \Ysp \cdots$,
is said to be \textit{jointly stationary} with $\rvax_{1}^{\infty}$ if and only if, for every $\ell\in\Nl$, the distribution of 
$(\rvax_{k}^{k+\ell},\rvay_{k}^{k+\ell})$ does not depend on $k$, for $k=1,2,\ldots$. 

The next theorem shows that, for such one-sided processes, joint-stationarity and causality may hold together only if $\rvay(k)$ is independent of $\rvax_{1}^{k-1}$ when $\rvax(k)$ is given.

\begin{thm}\label{thm:conditions_forJS_and_causal}
 If  $\rvax_{1}^{\infty}$ and $\rvay_{1}^{\infty}$ are jointly stationary and 
 $\rvay_{1}^{\infty}$ is
 causally related to 
 $\rvax_{1}^{\infty}$
 according to~\eqref{eq:MC_causality_0},
 then 
 \begin{align}\label{eq:MC_for_stat_and_causality0}
 \rvax_{k+1}^{\infty}
 \longleftrightarrow 
 \rvax(k)
 \longleftrightarrow 
 \rvay(k),
 \fspace 
 \forall k\in\set{1,2,\ldots}.
 \end{align}
 \finenunciado
\end{thm}

\begin{proof}
If~\eqref{eq:MC_for_stat_and_causality0} does not 
hold for some $k$ and if $\rvay_{1}^{\infty}$ and $\rvax_{1}^{\infty}$ are jointly stationary, then 
\begin{align}\label{eq:MC_false}
 \rvax_{2}^{\infty} 
 \longleftrightarrow
  \rvax(1)
 \longleftrightarrow
 \rvay(1)
 \end{align}
does not hold,
which corresponds to not satisfying~\eqref{eq:MC_causality_0} for $k=1$, completing the proof.
\end{proof}
To illustrate how restrictive condition~\eqref{eq:MC_for_stat_and_causality0} is, the next theorem  shows that, for a Gaussian $\kappa$-th order Markovian stationary source $\rvax_{1}^{\infty}$, causality and joint stationarity is possible only if $\rvax_{1}^{\infty}$ is i.i.d. ($\kappa=0$) or $\kappa=1$. 
Recall that a random (vector or scalar valued) process $\rvax_{1}^{\infty}$ is $\kappa$-th order Markovian if  $\kappa$ is the smallest non-negative integer such that
\begin{align}\label{eq:Markovian_source_def}
\rvax_{1}^{i}
\longleftrightarrow
\rvax_{i+1}^{i+k}
\longleftrightarrow
  \rvax_{i+k+1}^{\infty} 
  ,\fspace 
  \forall k\geq \kappa.
\end{align}
 
\begin{thm}\label{thm:stationarity_and_causality_Gaussian}
 Suppose $\rvax_{1}^{\infty}$ is a zero-mean  Gaussian stationary process, and  
 assume that, for some $N\in\set{2,3,\ldots}$, $\rvax_{1}^{N},\rvay_{1}^{N}$ are jointly Gaussian and jointly stationary,
 with 
 $\rvay_{1}^{N}$ being causally related to $\rvax_{1}^{N}$ according to~\eqref{eq:MC_causality_0}. 
 Then $\rvax_{1}^{\infty}$ is $\kappa$-th order Markovian with
 $\kappa\leq 1$.
\finenunciado
\end{thm}

\begin{proof}
Since $\rvax_{1}^{N}$ and $\rvay_{1}^{N}$ are jointly Gaussian and the latter depends causally upon the former, it holds that
\begin{align}
  \bK_{\rvey^{1}_{N}\rvex^{1}_{N}}\eq \Expe{[\rvay(1) \;\cdots\;\rvay(N)]^{T}[\rvax(1) \;\cdots\;\rvax(N)]}
  =\bA\bK_{\rvex^{1}_{N}}
  \label{eq:Kyx_A_and_Kx}
\end{align}
for some lower triangular matrix $\bA\in\Rl^{N\times N}$ having entries
$a_{i,j}\eq [\bA]_{i,j}, \; i,j\in\set{1,\ldots,N}$.
On the other hand, the fact that $\rvax_{1}^{N}$ and $\rvay_{1}^{N}$ are jointly stationary implies that 
$\bK_{\rvey^{1}_{N}\rvex^{1}_{N}}$ 
and 
$\bK_{\rvex^{1}_{N}}$
are Toeplitz matrices.
From~\eqref{eq:Kyx_A_and_Kx}, considering the entries on the first and second rows of 
$  \bK_{\rvey^{1}_{N}\rvex^{1}_{N}}$ 
and defining 
$$
\varrho_{k}\eq \Expe{\rvax(1)\rvax(1+k)},\fspace k=0,1\ldots,N-1,
$$ 
this Toeplitz condition implies that 
\begin{align*}
  a_{1,1}[\varrho_{0}\cdots \varrho_{N-2}]
  &=
  a_{2,1}[\varrho_{1}\cdots \varrho_{N-1}]
  +
  a_{2,2}[\varrho_{0}\cdots \varrho_{N-2}]
  \\
  \iff
  \underbrace{\frac{a_{1,1}-a_{2,2}}{ a_{2,1}}}_{\eq\zeta}
  [\varrho_{0}\cdots \varrho_{N-2}]
  &=
 [\varrho_{1}\cdots \varrho_{N-1}].
 \end{align*}
Therefore,
$  \varrho_{k}= \zeta\varrho_{k-1},  k=1,\ldots,N-1$,
which for a Gaussian stationary sequence $\rvax_{1}^{N}$
implies that 
$
\expe{\rvax(k)|\rvax_{1}^{k-1}}
=
\zeta \rvax(k-1)
$, $k=1,\ldots,N$.
For Gaussian random variables the latter is equivalent to the Markov chains
$
\rvax_{1}^{k-2}
\longleftrightarrow 
\rvax(k-1)
\longleftrightarrow 
\rvax(k)$, which defines 
a 1-st order Markovian process (if $\zeta\neq 0$) or an i.i.d. process (if $\zeta=0$).
This completes the proof.
\end{proof}

In the next section we will see that 
if $\rvax_{1}^{\infty}$ is $\kappa$-th order Markovian, then it is possible to build a pair $(\rvax_{1}^{\infty},\rvay_{1}^{\infty})$ causally related according to~\eqref{eq:MC_causality_0} such that $(\rvax_{\kappa}^{\infty},\rvay_{\kappa}^{\infty})$ is stationary.
Moreover, we will show in Theorem~\ref{thm:QJS_is_sufficient} below that the minimization associated with the causal IRDF can be restricted to such pairs.

\section{The Set of Quasi-Jointly Stationary Realizations is Sufficient}\label{sec:quasi_stat_is_sufficient}
In this section we show that for any $\kappa$-th order Markovian one-sided stationary source $\rvax_{1}^{\infty}$
the search for the causal IRDF (as defined in~\eqref{eq:Causal_IRDF_base} and for a large class of distortion criteria) can be restricted to output sequences $\rvay_{1}^{\infty}$ causally related to the source, 
jointly stationary with it after $\kappa$ samples, and such that 
$
\bar{I}(\rvax_{\kappa}^{\infty};\rvay_{\kappa}^{\infty})
=
\bar{I}(\rvax_{1}^{\infty};\rvay_{1}^{\infty})
$.
We refer to such pairs of processes as being \textit{quasi-jointly stationary} ($\kappa$-QJS), and define the set which contains them as follows:
\begin{defn}[Set of quasi-jointly stationary process]\label{def:quasi_jointly_stat}
The set of
$\kappa$-QJS distributions $\Qsp_{\kappa}$ is composed of all  joint distributions 
$P_{\rvax_{1}^{\infty},\rvay_{1}^{\infty}}$
of pairs of one-sided random processes $(\rvax_{1}^{\infty},\rvay_{1}^{\infty})$ which satisfy
\begin{align*}
(\rvax_{\kappa}^{\infty},\rvay_{\kappa}^{\infty}) &\text{ are jointly stationary}
  \\
  \lim_{\ell\to\infty}
  \frac{1}{\ell}I(\rvax_{1}^{\ell};\rvay_{1}^{\ell})
  &=
    \lim_{\ell\to\infty}
  \frac{1}{\ell}I(\rvax_{\kappa}^{\kappa+\ell-1};\rvay_{\kappa}^{\kappa+\ell-1}).\\[-14mm]
\end{align*}
\finenunciado  
\end{defn}

Notice that $\Qsp_{1}$ corresponds to the set of joint distributions associated with all jointly stationary one-sided process pairs.

As in~\cite{gorpin73}, we write $(\rvax_{1}^{k},\rvay_{1}^{k})\in (\Wsp_{D}^{1,k})$ when the distribution of 
$(\rvax_{1}^{k},\rvay_{1}^{k})$ belongs to the distortion-feasible set 
$\Wsp_{D}^{1,k}$, defined as in~\eqref{eq:W_D_k_k_right_marginal_cond}.

One can define a distortion-feasible set for pairs of one-sided processes 
$(\rvax_{1}^{\infty},\rvay_{1}^{\infty})$, say $(\Pspa_{D}^{\infty})$, from the finite-length distortion-feasible sets  $\set{\Wsp_{D}^{k,\ell}}_{k\leq \ell\in\Nl}$, in more than one manner.
A minimal condition we shall require for such definition is the following.
\begin{assu}\label{assu:Pspa_D_inf_assu}
The distortion-feasible set of distributions for pairs of one-sided processes $\Pspa_{D}^{\infty}$ satisfies the following:
\begin{enumerate}[i)]
\item If 
$(\tilde{\rvax}_{1}^{\infty},\tilde{\rvay}_{1}^{\infty})\in(\Pspa_{D}^{\infty})$, 
then $\tilde{\rvax}_{1}^{\infty}$ has the given probability distribution of the source process, say $P_{\mathring{\rvax}_{1}^{\infty}}$.
That is, $\Pspa_{D}^{\infty}\subset \Pspa^{1,\infty}$ (see~\eqref{eq:Pkk_def}).
\item If 
$(\tilde{\rvax}_{1}^{\infty},\tilde{\rvay}_{1}^{\infty})$ is any given pair of one-sided processes, and  
there exists an infinite collection of increasing integers $1=k_{1}<k_{2}<\cdots$ such that, 
for all $i\in\Nl$,
$(\tilde\rvax_{k_{i}}^{k_{i+1}-1},\tilde\rvay_{k_{i}}^{k_{i+1}-1})\in(\Wsp_{D}^{k_{i},k_{i+1}-1})$, then  $(\tilde\rvax_{1}^{\infty},\tilde\rvay_{1}^{\infty})\in(\Pspa_{D}^{\infty})$.
 \item For any pair of sequences 
$(\hat{\rvax}_{1}^{\ell},\hat{\rvay}_{1}^{\ell})\in(\Wsp_{D}^{1,\ell})$,  $\ell\in\Nl$,
and if $(\tilde\rvax_{1}^{\infty},\tilde\rvay_{1}^{\infty})\in(\Pspa_{D}^{\infty})$, 
then the concatenated processes
$\ddot{\rvax}_{1}^{\infty}\eq \set{\hat{\rvax}(1),\ldots,\hat{\rvax}(\ell),\tilde\rvax(1),\tilde\rvax(2),\cdots}$,
$\ddot{\rvay}_{1}^{\infty}\eq \set{\hat{\rvay}(1),\ldots,\hat{\rvay}(\ell),\tilde\rvay(1),\tilde\rvay(2),\cdots}$
satisfy $(\ddot{\rvax}_{1}^{\infty},\ddot{\rvay}_{1}^{\infty})\in(\Pspa_{D}^{\infty})$. 
\end{enumerate}
\finenunciado
\end{assu}
Notice that if $\Pspa_{D}^{\infty}$ satisfies this assumption and if the integers $k_{i}$ in Definition~\ref{def:W_D} were restricted to be positive, then we would have 
$\Wsp_{D}\subset \Pspa_{D}^{\infty}$ (see Definition~\ref{def:W_D}).
However, the one-way implications in Assumption~\ref{assu:Pspa_D_inf_assu} allow $\Pspa_{D}^{\infty}$ to be larger than $\Wsp_{D}$. 

We now define the sets of causally related pairs of sequences and processes.

\begin{defn}[Set of Causal Distributions]\label{def:Csp}
Define 
$(\Csp^{n})$ as
the set of all 
one-sided random $n$-sequences
$(\rvax_{1}^{n},\rvay_{1}^{n})$
which satisfy the causality constraint
\begin{align}\label{eq:MC_causality_base_def}
  \rvax_{k+1}^{n}
 \longleftrightarrow 
 \rvax_{1}^{k}
 \longleftrightarrow 
 \rvay_{1}^{k} 
 ,\fspace 
 k=1,\ldots,n
\end{align}
The set of causally related one-sided process pairs 
$(\Csp^{\infty})$ is defined likewise but for one-sided processes
$(\rvax_{1}^{\infty},\rvay_{1}^{\infty})$
which satisfy the causality constraint~\eqref{eq:MC_causality_0}. 
\end{defn}

With the above minimal notions, one can define two causal IRDFs, namely
\begin{align}
 R_{c}^{it}(D)
 &\eq \inf_{(\rvax_{1}^{\infty},\rvay_{1}^{\infty})\in(\Pspa_{D}^{\infty})\cap(\Csp^{\infty})} 
 \lim_{n\to\infty}
 \frac{1}{n}
 I(\rvax_{1}^{n};\rvay_{1}^{n})\label{eq:RcitD_for_infty}
\\
 \hat{R}_{c}^{it}(D)
 &\eq \lim_{n\to\infty}\inf_{(\rvax_{1}^{n},\rvay_{1}^{n})\in(\Wsp_{D}^{1,n})\cap(\Csp^{n})} \frac{1}{n}
 I(\rvax_{1}^{n};\rvay_{1}^{n}) \label{eq:RcitD_lim_inf}.
\end{align}
provided the limits exist.
The ``$\lim\inf$''  causal IRDF $\hat{R}_{c}^{it}(D)$  coincides with $\overline{H}^{0}_{D}$ if in~\eqref{eq:HD0_def} one fixes $k_{1}=1$ and lets $k_{2}\to\infty$.
By contrast, $R_{c}^{it}(D)$ differs from 
$\overrightarrow{H}^{0}_{D}$ in that the latter is associated with the (less general) distortion-feasible set $\Wsp_{D}$ (see Definition~\ref{def:W_D}).

Since our main result will be stated with the assumption that 
$R_{c}^{it}(D)= \hat{R}_{c}^{it}(D)$, we develop next two sufficient conditions for such equality to hold.

\subsection{Sufficient Conditions for $R_{c}^{it}(D)= \hat{R}_{c}^{it}(D)$}\label{sec:suf_cond_for_inflim_eq_liminf}
We begin by stating a useful construction of a pair of processes from a finite-length sequence and some if the properties of the former.
\footnote{A similar construction is introduced in the proof of~\cite[Lemma~1]{gorpin73}, for two-sided processes. 
Here we modify this idea to construct one-sided processes and reveal some of their properties.}
\begin{prop}\label{prop:construction}
Let $(\breve{\rvax}_{1}^{\infty},\breve{\rvay}_{1}^{n})$,   be given, with $\breve{\rvax}_{1}^{\infty}$  stationary and such that 
$(\breve{\rvax}_{1}^{\infty},\breve{\rvay}_{1}^{n})$ satisfies the causality condition~\eqref{eq:MC_causality_0} for $k=1,2,\ldots,n$.
Build the processes $(\tilde{\rvax}_{1}^{\infty},\tilde{\rvay}_1^{\infty})$ as follows:
\begin{subequations}\label{subeq:tilde_xy_def}
 \begin{enumerate}
\item[i)] Choose $\tilde{\rvax}_{1}^{\infty}$ with the same distribution of $\breve{\rvax}_{1}^{\infty}$, i.e.,
\begin{align}
 \tilde{\rvax}_{1}^{\infty}\sim  \breve{\rvax}_{1}^{\infty}.
\end{align}
\item[ii)] For every $\ell\in\Nl_{0}$, choose the conditional distribution 
of 
$\tilde{\rvay}_{n\ell+1}^{n(\ell+1)}$
given $\tilde{\rvax}_{1}^{\infty},\tilde{\rvay}_{1}^{n\ell} ,\tilde{\rvay}_{n(\ell+1)+1}^{\infty} $ as
\begin{align}
\label{eq:same_conditional_distributions0}
P_{\tilde{\rvay}_{n\ell+1}^{n(\ell+1)}
|
\tilde{\rvax}_{1}^{\infty},\tilde{\rvay}_{1}^{n\ell} ,\tilde{\rvay}_{n(\ell+1)+1}^{\infty} }
=
 P_{\tilde{\rvay}_{n\ell+1}^{n(\ell+1)}|\tilde{\rvax}_{n\ell+1}^{n(\ell+1)}}
 =
 P_{\breve{\rvay}_{1}^{n}|\breve{\rvax}_{1}^{n}}.
\end{align}
\end{enumerate}
\end{subequations}
Then
\begin{align}
\label{eq:i_same0}
I(\tilde{\rvax}_{n\ell+1}^{n\ell+n};\tilde{\rvay}_{n\ell+1}^{n\ell+n}) = I(\breve{\rvax}_{1}^{n};\breve{\rvay}_{1}^{n}), \fspace \forall \ell \in \Nl_0.
\end{align}
\begin{align}
\frac{1}{kn} I(\tilde{\rvax}_1^{ \ell n};\tilde{\rvay}_1^{\ell n})
\leq
\frac{1}{n} I(\breve{\rvax}_{1}^{n};\breve{\rvay}_{1}^{n}), \fspace \forall \ell \in \Nl.
\label{eq:Irate_of_block_less_than_sequence_of_blocks0}
\end{align}
Also:
\begin{enumerate}
 \item[a)] $ $
If 
$(\breve{\rvax}_{1}^{n},\breve{\rvay}_{1}^{n}) \in (\Wsp_D^{1,n})$, 
Assumption~\ref{assu:Pspa_D_inf_assu} and
\begin{align}\label{eq:W_D_are_stationary0}
\forall k,\ell\in\Nl, t\in\Nl_{0}:\fspace  \Wsp_{D}^{k,k+t} = \Wsp_{D}^{\ell,\ell+t},
\end{align}
hold, then 
$(\tilde{\rvax}_1^{\infty},\tilde{\rvay}_1^{\infty}) \in (\Pspa_D^{\infty})$.

\item[b)] If  $\breve{\rvax}_{1}^{\infty}$ is $\kappa$-th order Markovian, then 
\begin{align}\label{eq:MC_causality_kappa}
 \tilde{\rvax}_{m+k+1}^{\infty}
 \longleftrightarrow
 \tilde{\rvax}_{m+2-\kappa}^{m+k}
 \longleftrightarrow
 \tilde{\rvay}_{m+1}^{m+k},
 \fspace k\in\Nl, m\in\Nl_{0}.
\end{align}
\end{enumerate}
\finenunciado
\end{prop}

\begin{proof}
The first equality in~\eqref{eq:same_conditional_distributions0} is equivalent to the Markov chains
\begin{align}
\label{eq:MC_for_tilde_y}
\tilde{\rvay}_{n\ell+1}^{n\ell+n}
\longleftrightarrow
\tilde{\rvax}_{n\ell+1}^{n\ell+n}
\longleftrightarrow
(\tilde{\rvax}_{-\tau}^{0},\tilde{\rvax}_{ni+1}^{ni+n},\tilde{\rvay}_{ni+1}^{ni+n}), \fspace 
\ell,i \in \Nl_0, i\neq \ell.
\end{align}
The second equality in~\eqref{eq:same_conditional_distributions0} together with the fact that $\tilde{\rvax}_1^{\infty}$ is stationary imply that $(\tilde{\rvax}_{\ell n+1}^{\ell n+n},\tilde{\rvay}_{\ell n+1}^{\ell n+n}) \sim (\breve{\rvax}_{1}^{n},\breve{\rvay}_{1}^{n})$, and thus~\eqref{eq:i_same0} follows. 
On the other hand, we have that
\begin{align}
\frac{1}{\ell n} I(\tilde{\rvax}_1^{\ell n};\tilde{\rvay}_1^{\ell n})
\leq
\frac{1}{\ell n} \sum_{i=0}^{\ell -1} 
I(\tilde{\rvax}_{ni+1}^{ni+n};\tilde{\rvay}_{ni+1}^{ni+n})
\overset{\eqref{eq:i_same0}}{=}
\frac{1}{n} I(\breve{\rvax}_{1}^{n};\breve{\rvay}_{1}^{n}),
\label{eq:Irate_of_block_less_than_sequence_of_blocks}
\end{align}
where the first inequality holds due to Proposition~\ref{prop:I_for_MC}, in the Appendix (applying it successively to the sequences 
$(\tilde{\rvax}_1^{\ell n},\tilde{\rvay}_1^{\ell n})$, $\ell\in\Nl$, which can be done because they satisfy \eqref{eq:MC_for_tilde_y}). 
This proves~\eqref{eq:Irate_of_block_less_than_sequence_of_blocks0}.

The fact that $(\breve{\rvax}_{1}^{n},\breve{\rvay}_{1}^{n}) \in (\Wsp_D^{1,n})$,
the definition of $(\tilde{\rvax}_1^{\infty},\tilde{\rvay}_1^{\infty})$ 
and the stationarity condition~\eqref{eq:W_D_are_stationary0}
imply that $(\tilde{\rvax}_{n\ell +1}^{n\ell +n},\tilde{\rvay}_{n\ell +1}^{n\ell +n})\in(\Wsp_{D}^{n\ell +1,n\ell +n})$, for all $\ell \in \Nl_0$. 
The latter together with 
Assumption~\ref{assu:Pspa_D_inf_assu}
implies that $(\tilde{\rvax}_1^{\infty},\tilde{\rvay}_1^{\infty}) \in (\Pspa_D^{\infty})$.

On the other hand,~\eqref{eq:MC_for_tilde_y}  together with the fact that
$(\breve{\rvax}_{1}^{\infty},\breve{\rvay}_{1}^{n})$ satisfies~\eqref{eq:MC_causality_0} for $k=1,2,\ldots,n$, generalizes to
\begin{align}
\label{eq:MC_for_tilde_y_general}
\tilde{\rvay}_{n\ell+1}^{n\ell + {k}}
\longleftrightarrow
\tilde{\rvax}_{n\ell+1}^{n\ell+{k}}
\longleftrightarrow
(\tilde{\rvax}_{n\ell+{k}+1}^{n\ell+n},
\tilde{\rvax}_{ni+1}^{ni+n},\tilde{\rvay}_{ni+1}^{ni+n}),
\fspace \ell,i \in \Nl_0, i\neq \ell,
{k}\in\Nl.
\end{align}
The latter implies the Markov chains 
\begin{align}
  \label{eq:MC_for_tilde_y_short}
  \underbrace{\tilde{\rvay}_{n\ell+m}^{n\ell + {k}}}_{\rvab}
\longleftrightarrow
\underbrace{\tilde{\rvax}_{n\ell+1}^{n\ell+{k}}}_{\rvac,\rvad}
\longleftrightarrow
\underbrace{\tilde{\rvax}_{n\ell+{k}+1}^{\infty}}_{\rvaa},
\fspace  \ell \in \Nl_0,
1\leq m\leq k\in\Nl.
\end{align}
Supposing $\breve{\rvax}_{1}^{n}$ is $\kappa$-th order Markovian, it holds that 
\begin{align}\label{eq:MC_markovity_ofx_short}
\underbrace{\tilde{\rvax}_{n\ell+1}^{n\ell+k-\kappa}}_{\rvad}
  \longleftrightarrow
\underbrace{\tilde{\rvax}_{n\ell+k+1-\kappa}^{n\ell+{k}}}_{\rvac}
  \longleftrightarrow
  \underbrace{\tilde{\rvax}_{n\ell+{k}+1}^{\infty}}_{\rvaa},
\fspace \ell \in \Nl_0, 
  k\in\Nl.
\end{align}
Invoking Proposition~\ref{prop:abcd} with $\rvaa,\rvab,\rvac,\rvad$ according to the labeling in~\eqref{eq:MC_for_tilde_y_short} and in~\eqref{eq:MC_markovity_ofx_short}, we readily obtain that
$(\rvab,\rvad)\longleftrightarrow 
\rvac 
\longleftrightarrow 
\rvaa$, implying
\begin{align}
  \tilde{\rvay}_{n\ell+m}^{n\ell + {k}} 
\longleftrightarrow
 \tilde{\rvax}_{n\ell+m+1-\kappa}^{n\ell+{k}} 
\longleftrightarrow
 \tilde{\rvax}_{n\ell+{k}+1}^{\infty},
\fspace  \ell \in \Nl_0,
1\leq m\leq k\in\Nl.
\end{align}
Thus $(\tilde{\rvax}_{1}^{\infty},\tilde{\rvay}_1^{\infty})$ satisfies the causality condition~\eqref{eq:MC_causality_kappa}.
This completes the proof.
\end{proof}

We now state a technical lemma which is akin to~\cite[Theorem~2]{gorpin73} but for one-sided processes, the proof of which can be found in Appendix~\ref{prof:lem_liminf_leq_inflim}.
\begin{lem}\label{lem:liminf_leq_inflim}
Let $\mathring{\rvax}_{1}^{\infty}$ be a stationary one-sided source 
and suppose the distortion-feasible sets   
$\set{\Wsp_{D}^{k,\ell}}_{k\leq\ell\in\Nl}$
and  $\Pspa_{D}^{\infty}$ satisfy Assumption~\ref{assu:Pspa_D_inf_assu} and the stationarity condition
\begin{align}\label{eq:W_D_are_stationary}
\forall k,\ell\in\Nl, t\in\Nl_{0}:\fspace  \Wsp_{D}^{k,k+t} = \Wsp_{D}^{\ell,\ell+t}.
\end{align}
Then
\begin{align}
 R_{c}^{it}(D)\leq \hat{R}_{c}^{it}(D).
\end{align}
\finenunciado
\end{lem}

Next, we propose a possible definition of $\Pspa_{D}^{\infty}$ general enough to encompass  the asymptotic single-letter fidelity criteria described by~\eqref{eq:asymptotic_distortion_constraint0}.
For that purpose, 
we need to define 
\begin{align}
 \Pspa^{1,\infty} \eq 
 \Set{P_{\dot{\rvax}_{1}^{\infty},\dot{\rvay}_{1}^{\infty}}:
 P_{\dot{\rvax}_{k}^{\ell},\dot{\rvay}_{k}^{\ell}}\in\Pspa^{k,\ell}, 
 \quad \forall k,\ell\in\Nl, \ 
 k\leq \ell
 }
\end{align}
and require the distortion-feasible sets to satisfy the following assumption:

\begin{assu}\label{assu:DFS_as_ineq_for_k}
 The distortion-feasible sets $\set{\Wsp_{D}^{k,\ell}}_{k\leq \ell\in\Nl}$ can be expressed as
 \begin{align}\label{eq:Wsp_k_l_assu}
  \Wsp_{D}^{k,\ell}= 
  \Set{P_{\rvax_{k}^{\ell},\rvay_{k}^{\ell}}  
  :
  \rho^{k,\ell}(P_{\rvax_{k}^{\ell},\rvay_{k}^{\ell}})\leq D
  }
 \end{align}
for some non-negative distortion maps 
$\rho^{k,\ell} :\Pspa^{k,\ell}\to \Rl_{0}^{+}
$, $k\leq \ell\in\Nl$.
Moreover, 
the distortion-feasible set for one-sided sequences, $\Pspa_{D}^{\infty}$, has the form 
\begin{align}\label{eq:Pdinfty_def}
 \Pspa_{D}^{\infty} = \Set{P_{\rvax_{1}^{\infty},\rvay_{1}^{\infty}} \in\Pspa^{1,\infty}
 :
 \limsup_{k\to\infty}\rho^{1,k}(P_{\rvax_{1}^{k},\rvay_{1}^{k}})\leq D
 }
\end{align}

\finenunciado
\end{assu}
Notice that with such construction, $\Pspa_{D}^{\infty}$ does not necessarily satisfy Assumption~\ref{assu:Pspa_D_inf_assu}.
Also, the distortion-feasible sets $\set{\Wsp_{D}^{k,\ell}}_{k\leq \ell\in\Nl}$ with the specific form given by~\eqref{eq:Wsp_k_l_assu} do not necessarily satisfy the stationarity condition~\eqref{eq:W_D_are_stationary}.

This definition, based on the limit of a sequence of  distortion functions, is clearly capable of representing the general asymptotic single-letter criteria of~\eqref{eq:asymptotic_distortion_constraint0} while satisfying Assumption~\ref{assu:Pspa_D_inf_assu}. 
Recall that, as shown in Lemma~\ref{lem:Wsp_D_cannot_represent_asymptotic_SLDC}, it is not possible to do this with the distortion-feasible set $\Wsp_{D}$ from~\cite{gorpin73}, given by Definition~\ref{def:W_D}. 
In addition, the construction of $\Pspa_{D}^{\infty}$ provided by~\eqref{eq:Pdinfty_def}
allows for several specific criteria commonly found in the literature, such as
the one utilized in~\cite{derost12} and in the definition of a rate-distortion achievable pair in~\cite[p.~306]{covtho06}.

We are now in the position to provide two independent conditions that are sufficient to ensure $ R_{c}^{it}(D)=\hat{R}_{c}^{it}(D)$ (the proof is given in Appendix~\ref{proof:liminf_equals_inflim}).

\begin{lem}\label{lem:liminf_equals_inflim}
Consider the same conditions given in the statement of Lemma~\ref{lem:liminf_leq_inflim}.
If, in addition,  any of  the two following conditions holds
\begin{enumerate}[i)]
 \item For every $D>0$, there exists $N<\infty$ such that, 
 \begin{align}\label{eq:weak_consistencia}
\forall n\geq N:\fspace 
 (\rvax_{1}^{\infty},\rvay_{1}^{\infty})\in (\Pspa_{D}^{\infty} )
 \Longrightarrow
 (\rvax_{1}^{n},\rvay_{1}^{n})\in (\Wsp_{D}^{1,n}) 
 \end{align}
\item $R_{c}^{it}(D)$ is continuous and Assumption~\ref{assu:DFS_as_ineq_for_k} holds,
 \end{enumerate}
then
\begin{align}\label{eq:RcitD_equals_Rcitd_hat}
 R^{it}_{c}(D) = \hat{R}_{c}^{it}(D).
\end{align} 
\finenunciado
\end{lem}

\subsection{Main Result}
With the above notions, we can state the main result of this section, akin to 
Theorem~\ref{thm:4_gorpin} but for one-sided processes and for the corresponding causality condition given by~\eqref{eq:MC_causality_0} (the proof is presented in Appendix~\ref{proof:QJS_sufficient}):

\begin{thm}\label{thm:QJS_is_sufficient}
 Let the source $\mathring{\rvax}_{1}^{\infty}$ be a one-sided stationary $\kappa$-th order Markovian process 
 and suppose that 
 $\hat{R}^{it}(D)=R_{c}^{it}(D)$, where $R_{c}^{it}(D)$ and $\hat{R}^{it}(D)$ are as defined in~\eqref{eq:RcitD_for_infty} and~\eqref{eq:RcitD_lim_inf}, respectively. 
Furthermore, suppose that the distortion-feasible sets 
$\Pspa_{D}^{\infty}$, 
$\set{\Wsp_{D}^{k,\ell}}_{k\leq\ell\in\Z}$ satisfy Assumption~\ref{assu:Pspa_D_inf_assu} and 
\begin{enumerate}
\item The ``shift-invariance'' condition:
\begin{align}\label{eq:shift_invariant_cond}
\forall L\in\Nl:\fspace 
 (\procx,\procy)\in (\Pspa_{D}^{\infty}) 
 \Longrightarrow
 \frac{1}{L}\Sumfromto{\ell=1}{L} P_{\rvax_{\ell}^{\infty},\rvay_{\ell}^{\infty}}
 \in\Pspa_{D}^{\infty}.
\end{align}

\item  The stationarity condition given by~\eqref{eq:W_D_are_stationary}.

\item The ``first-samples condition'':
There exists a pair of random sequences 
$(\rvax_{1}^{\kappa-1},\dot{\rvay}_{1}^{\kappa-1})\in(\Wsp_{D}^{1,\kappa-1})\cap(\Csp^{\kappa-1})$ 
and such that $I(\rvax_{1}^{\kappa-1};\dot{\rvay}_{1}^{\kappa-1})<\infty$.
\end{enumerate}
Then the minimization in the definition of $R_{c}^{it}(D)$ in~\eqref{eq:RcitD_for_infty} can be restricted to pairs of
processes $(\rvax_{1}^{\infty},\rvay_{1}^{\infty})$ with distributions in
$\Qsp_{\kappa}\cap\Pspa_{D}^{\infty}\cap \Csp^{\infty}$ which, in addition, satisfy 
$(\rvax_{\kappa}^{\infty},\rvay_{\kappa}^{\infty})\in(\Pspa_{D}^{\infty})$.
\finenunciado
\end{thm}
\vspace{2mm}

Putting aside the obvious difference between Theorem~\ref{thm:QJS_is_sufficient} and Theorem~\ref{thm:4_gorpin} arising from the fact that the former considers as sources one-sided processes and the latter two-sided processes, it is worth drawing a parallel between these two theorems.
The requirement of Assumption~\ref{assu:Pspa_D_inf_assu} in Theorem~\ref{thm:QJS_is_sufficient} is weaker than the requirement of $\Wsp_{D}$ to conform to Definition~\ref{def:W_D} in  Theorem~\ref{thm:4_gorpin}.
Thus, Theorem~\ref{thm:QJS_is_sufficient} holds for a larger class of fidelity criteria.
The assumption that $\hat{R}^{it}(D)=R_{c}^{it}(D)$ can be seen as the equivalent of condition~iv) in Theorem~\ref{thm:4_gorpin} translated to the setting of one-sided sequences.
The same is true with condition~2) in Theorem~\ref{thm:QJS_is_sufficient} with respect to condition~ii) in Theorem~\ref{thm:4_gorpin}.
However, no conditions are stated in~\cite{gorpin73} which suffice for condition~iv) in Theorem~\ref{thm:4_gorpin} to hold.
In contrast, we have provided Lemma~\ref{lem:liminf_equals_inflim}, which gives two independent conditions under which $\hat{R}^{it}(D)=R_{c}^{it}(D)$ is satisfied.  
The other assumptions in Theorem~\ref{thm:QJS_is_sufficient} differ from those in Theorem~\ref{thm:4_gorpin}.
Condition 1) in Theorem~\ref{thm:QJS_is_sufficient} is weaker than condition~v) in Theorem~\ref{thm:4_gorpin}.
Condition 3) in Theorem~\ref{thm:QJS_is_sufficient} is absent in Theorem~\ref{thm:4_gorpin}, and is required in our proof as a consequence of the transient behavior arising from treating one-sided processes (see Theorem~\ref{thm:conditions_forJS_and_causal} in Section~\ref{sec:joint_stat_and_causality}). 

\begin{rem}
 Among the distortion criteria which satisfy the conditions of Theorem~\ref{thm:QJS_is_sufficient}, we find the family of asymptotic single-letter constraints of~\eqref{eq:asymptotic_distortion_constraint0} (by letting 
 $
 \Wsp_{D}^{k,\ell} = 
 \set{ P_{\dot{\rvax}_{k}^{\ell},\dot{\rvay}_{k}^{\ell}}: 
 P_{\dot{\rvax}_{k}^{\ell}}=P_{\rvax_{k}^{\ell}},
 (\ell-k+1)^{-1}\sumfromto{i=k}{\ell}\rho(P_{\rvax(i),\rvay(i)})\leq D}
 $).
 Recall that this class of distortion criteria cannot be expressed using a distortion-feasible set $\Wsp_{D}$ conforming to Definition~\ref{def:W_D}, and hence it is not covered by Theorem~\ref{thm:4_gorpin}.
 \finenunciado
\end{rem}

As pointed out by Remark~\ref{rem:strong_causality} in the Introduction, if now one supposes that $\mathring{\rvax}_{1}^{\infty}$ is the positive-time part of a two-sided stationary process $\mathring{\rvax}_{-\infty}^{\infty}$, then the fact that
$(\mathring{\rvax}_{1}^{\infty},\rvay_{1}^{\infty})\in (\Csp^{\infty})$ does not guarantee that 
$\rvay_{1}^{\infty}$ depends causally on $\mathring{\rvax}_{1}^{\infty}$.
For the latter to hold in this situation, it is required that $(\mathring{\rvax}_{-\infty}^{\infty},\rvay_{1}^{\infty})$ satisfies~\eqref{eq:MC_causality_hybrid}.
This implies that for this case, the definition of the causal IRDF as stated in~\eqref{eq:RcitD_for_infty} needs to be extended to 
\begin{align}
  R_{c\text{(strong)}}^{it}(D) \eq \inf_{
  \subalign{ (\rvax_{-\infty}^{\infty},\,\rvay_{1}^{\infty})\, :\, &  \rvax_{-\infty}^{\infty}\sim \,\mathring{\rvax}_{-\infty}^{\infty} \\
  & (\rvax_{-\infty}^{\infty},\,\rvay_{1}^{\infty}) \text{ satisfies \eqref{eq:MC_causality_hybrid}}\\
  & (\rvax_{1}^{\infty},\rvay_{1}^{\infty})\in (\Pspa_{D}^{\infty})
  }
  }
  \lim_{n\to\infty}
  \frac{1}{n}I(\rvax_{1}^{n};\rvay_{1}^{n}).
  \label{eq:Causal_RDF_strong}
\end{align}
Notice that when the source lacks a negative-time part,~\eqref{eq:MC_causality_hybrid} simplifies to~\eqref{eq:MC_causality_0}, and thus $R_{c\text{(strong)}}^{it}(D)$ becomes equal to $R_{c}^{it}(D)$.

The above observations raise the question of whether Theorem~\ref{thm:QJS_is_sufficient} can be extended for the case in which the one-sided source is the positive-time part of a two-sided stationary process.
This implies considering~$R_{c\text{(strong)}}^{it}(D)$ instead of $R_{c}^{it}(D)$, or, equivalently, the strong causality constraint~\eqref{eq:MC_causality_hybrid} instead of the short causality constraint~\eqref{eq:MC_causality_0}.

It turns out that Theorem~\ref{thm:QJS_is_sufficient} can indeed be extended for this situation, thanks to the 
following proposition, the proof of which can be found in Section~\ref{proof:of_prop:there_exits_ybar_strongly_causal}.

\begin{prop}\label{prop:there_exits_ybar_strongly_causal}
Suppose that $\rvax_{1}^{\infty}$ is the positive-time part of a two-sided process $\rvax_{-\infty}^{\infty}$ and that 
$(\rvax_{1}^{\infty},\rvay_{1}^{\infty})$ satisfies the one-sided causality condition~\eqref{eq:MC_causality_0}, i.e.,
\begin{align}\label{eq:short_causality}
 \rvax_{k+1}^{\infty} \longleftrightarrow
  \rvax_{1}^{k}
  \longleftrightarrow
  \rvay_{1}^{k},\fspace k\in\Nl.
\end{align}  
Then there exists (or, equivalently, one can construct) a one-sided random process $\bar{\rvay}_{1}^{\infty}$ 
such that 
\begin{align}
  (\rvax_{1}^{\infty},\bar{\rvay}_{1}^{\infty})
  \sim 
  (\rvax_{1}^{\infty}, \rvay_{1}^{\infty})\label{eq:likewise}
  \\
  (\rvax_{-\infty}^{0},\rvax_{k+1}^{\infty})
  \longleftrightarrow
  \rvax_{1}^{k}
  \longleftrightarrow
  \bar{\rvay}_{1}^{k}\label{eq:strong_causality_ybar}
\end{align}
\finenunciado
\end{prop}

Thanks to Proposition~\ref{prop:there_exits_ybar_strongly_causal}, we have the following extension of Theorem~\ref{thm:QJS_is_sufficient}.

\begin{thm}[Extension of Theorem~\eqref{thm:QJS_is_sufficient}]\label{thm:QJS_is_sufficient_strong}
 Let the source $\mathring{\rvax}_{1}^{\infty}$ be the positive-time part of a two-sided stationary $\kappa$-th order Markovian process~$\mathring{\rvax}_{-\infty}^{\infty}$.
 Under the same assumptions made in the statement of Theorem~\ref{thm:QJS_is_sufficient}, the infimization in the definition of $R_{c}^{it}(D)$ in~\eqref{eq:Causal_RDF_strong} can be restricted to pairs of processes 
 $(\rvax_{-\infty}^{\infty},\rvay_{1}^{\infty})$ which satisfy~\eqref{eq:MC_causality_hybrid} and such that  
 $(\rvax_{1}^{\infty},\rvay_{1}^{\infty})\in(\Qsp_{\kappa}\cap \Pspa_{D}^{\infty})$ and 
 $(\rvax_{\kappa}^{\infty},\rvay_{\kappa}^{\infty})\in (\Pspa_{D}^{\infty})$.
 Moreover, $R_{c\text{(strong)}}^{it}(D)= R_{c}^{it}(D)$. 
 \finenunciado
\end{thm}

\begin{proof}
 The only difference between the RHS of~\eqref{eq:RcitD_for_infty} and~\eqref{eq:Causal_RDF_strong} is that they consider the causality constraints~\eqref{eq:MC_causality_0} and~\eqref{eq:MC_causality_hybrid}, respectively. 
 First, notice that 
 \begin{align}
  (\rvax_{-\infty}^{\infty},\rvay_{1}^{\infty}) \text{ satisfies~\eqref{eq:MC_causality_hybrid} } 
  \Longrightarrow 
  (\rvax_{1}^{\infty},\rvay_{1}^{\infty}) \text{ satisfies~\eqref{eq:MC_causality_0} }
 \end{align}
and thus $R_{c\text{(strong)}}^{it}(D)\geq R_{c}^{it}(D)$.
On the other hand, from Theorem~\ref{thm:QJS_is_sufficient}, the infimization yielding $R_{c}^{it}(D)$ can be carried out considering only pairs of processes 
$(\rvax_{1}^{\infty},\rvay_{1}^{\infty})$ satisfying~\eqref{eq:MC_causality_0} and
$(\rvax_{1}^{\infty},\rvay_{1}^{\infty})\in(\Pspa_{D}^{\infty}\cap \Qsp_{\kappa})$, 
$(\rvax_{\kappa}^{\infty},\rvay_{\kappa}^{\infty})\in(\Pspa_{D}^{\infty})$.
But
as a consequence of Proposition~\ref{prop:there_exits_ybar_strongly_causal}, for every such 
$(\rvax_{1}^{\infty},\rvay_{1}^{\infty})$,
there exists a process $\bar{\rvay}_{1}^{\infty}$ such that  
$(\rvax_{-\infty}^{\infty},\bar{\rvay}_{1}^{\infty})$ satisfies~\eqref{eq:MC_causality_hybrid} and 
 \begin{align}
 (\rvax_{1}^{\infty},\bar{\rvay}_{1}^{\infty})&\in(\Pspa_{D}^{\infty}\cap \Qsp_{\kappa}), \label{eq:inPD_Qk}\\
 (\rvax_{\kappa}^{\infty},\bar{\rvay}_{\kappa}^{\infty})&\in(\Pspa_{D}^{\infty}),\label{eq:inPD}\\
 \bar{I}(\rvax_{1}^{\infty};\bar{\rvay}_{1}^{\infty}) &= \bar{I}(\rvax_{1}^{\infty}; \rvay_{1}^{\infty}).
\end{align}
This implies that the minimization associated with $R_{c\text{(strong)}}^{it}(D)$ can be restricted to pairs $(\rvax_{-\infty}^{\infty},\bar{\rvay}_{1}^{\infty})$ satisfying~\eqref{eq:MC_causality_hybrid},~\eqref{eq:inPD_Qk} and~\eqref{eq:inPD} (proving the first claim of the theorem) and 
that  
$R_{c\text{(strong)}}^{it}(D)\leq R_{c}^{it}(D)$.
The latter and the reverse inequality confirms that 
$R_{c\text{(strong)}}^{it}(D)= R_{c}^{it}(D)$, concluding the proof.
\end{proof}

\subsection{Correspondence Between $R_{c}^{it}(D)$ and $\overline{R}_{c}^{it}(D)$}
The purpose of this section is to establish a correspondence between $R_{c}^{it}(D)$ and $\overline{R}_{c}^{it}(D)$ (introduced in~\cite{derost12}).
As we discuss next, drawing an appropriate comparison between these two causal IRDFs requires two modifications to the definition of  $\overline{R}_{c}^{it}(D)$ already described on page~\pageref{eq:overlineRcitD_def}. 

The first modification consists of extending 
$\overline{R}_{c}^{it}(D)$ to account for arbitrary fidelity criteria embodied in an arbitrary distortion-feasible set 
$\Pspa_{D}^{\infty}\subset \Pspa^{1,\infty}$.

The second modification is necessary in order to make $\overline{R}_{c}^{it}(D)$ a tighter lower bound to the corresponding infimal operational data rate.
To see why, it is necessary to recall how 
$\bar{I}(\rvax_{1}^{\infty};\rvay_{1}^{\infty})$
lower bounds the operational data rate of encoding $\rvax_{1}^{\infty}$ and decoding it as $\rvay_{1}^{\infty}$.
For this purpose, let $\rveb(k)$ be the random binary sequence produced by the encoder from time $1$ to $k$ and let $\abs{\rveb(k)}$ be the length of $\rveb(k)$ (in bits).
Since the code must be uniquely decodable%
\footnote{If zero-delay operation is required, it must be an instantaneous code, which is a sub-class of uniquely decodable codes~\cite[\S~5.1]{covtho06}.}%
, the bit-string $\rveb(k)$ satisfies the Kraft inequality~\cite[\S~5.1]{covtho06}. 
In general, $\rveb(k)$ can be generated by using not only $\rvax_{1}^{k}$ but also  $\rvax_{-\infty}^{0}$,
 and thus
\begin{align}
 \Expe{\abs{\rveb(k)} }
 \overset{(a)}{\geq }
 H(\rveb(k))
 =
 I(\rveb(k);\rveb(k))
 \overset{(b)}{\geq }
 I(\rvax_{-\infty}^{k};\rvay_{1}^{k})
 \overset{(\textbf{P1})}{\geq }
 I(\rvax_{1}^{k};\rvay_{1}^{k}),
\end{align}
where~$(a)$ follows from~\cite[Theorem 5.3.1]{covtho06} and~$(b)$ is a consequence of the data-processing inequality~\cite[Theorem 2.8.1]{covtho06}.
Thus, 
$I(\rvax_{1}^{k};\rvay_{1}^{k})$ lower bounds the operational data rate 
$k^{-1} \Expe{\abs{\rveb(k)} }$
as tightly as 
 $k^{-1}I(\rvax_{-\infty}^{k};\rvay_{1}^{k})$ if and only if 
\begin{align}\label{eq:MC_only_from_1}
  \rvay_{1}^{k}
  \longleftrightarrow
  \rvax_{1}^{k}
  \longleftrightarrow
  \rvax_{-\infty}^{0}.
\end{align}
This Markov chain, combined with the causality constraint~\eqref{eq:MC_causality_2sdd} 
\begin{align}
  \rvay_{1}^{k}
  \longleftrightarrow
  \rvax_{-\infty}^{k}
  \longleftrightarrow
  \rvax_{k+1}^{\infty},
\end{align}
implies from Proposition~\ref{prop:abcd} (in the Appendix) that 
\begin{align}
  \rvay_{1}^{k}
  \longleftrightarrow
  \rvax_{1}^{k}
  \longleftrightarrow
  \rvax_{k+1}^{\infty},
\end{align}
which is precisely the causality constraint for one-sided sources~\eqref{eq:MC_causality_0}.
But, as we have shown in theorems~\ref{thm:conditions_forJS_and_causal} and~\ref{thm:stationarity_and_causality_Gaussian}, such causality constraint is, in general, incompatible with the joint stationarity of 
$(\rvax_{1}^{\infty},\rvay_{1}^{\infty})$.
As a consequence, since such joint-stationarity is required by $\overline{R}_{c}^{it}(D)$, Markov chain~\eqref{eq:MC_only_from_1} cannot hold.
This means that when the causality constraint for a two-sided source established by~\eqref{eq:MC_causality_2sdd} is imposed,
$
  \lim_{\ell,n\to\infty}\frac{1}{n}I(\rvax_{-\ell}^{n};\rvay_{1}^{n})
$
is a tighter lower bound to the operational data rate than $\bar{I}(\rvax_{1}^{\infty};\rvay_{1}^{\infty})$.

Following these observations, we propose here the following modified definition of $\overline{R}_{c}^{it}(D)$.
\begin{defn}[An Improved and More General Definition of $\overline{R}_{c}^{it}(D)$]\label{def:Overline_RCitd_redef}
For any given $\rvax_{-\infty}^{\infty}$ two-sided stationary  source, redefine the  causal stationary IRDF introduced in~\cite[Definition~6]{derost12} as
\begin{align}
 \overline{R}_{c}^{it}(D) \eq 
 \inf_{
\subalign{ 
 (\rvax_{-\infty}^{\infty},\rvay_{1}^{\infty}): 
 &(\rvax_{1}^{\infty},\rvay_{1}^{\infty})
 \in
 (\Pspa_{D}^{\infty}) \cap 
 (\Qsp_{1})
 \\
 &(\rvax_{-\infty}^{\infty},\rvay_{1}^{\infty})
 \in(\Csp_{-\infty}^{\infty})
 }
}
  \lim_{\ell\to-\infty} \lim_{n\to\infty}\frac{1}{n}I(\rvax_{\ell}^{n};\rvay_{1}^{n})
  \label{eq:Overline_RCitd_redef}
\end{align}
where 
$(\Qsp_{1})$ is the set of all pairs one-sided jointly stationary random processes, and 
$(\Csp_{-\infty}^{\infty})$ is the set of all pairs of random processes $(\rvax_{-\infty}^{\infty},\rvay_{1}^{\infty})$ which satisfy the causality constraint~\eqref{eq:MC_causality_2sdd} .
\end{defn}

We can now state the following corollary of Theorem~\ref{thm:QJS_is_sufficient}, the proof of which can be found in Appendix~\ref{proof:coro_overlineRcitD_equals_RcitD}:
\begin{coro}\label{coro:overlineRcitD_equals_RcitD}
 Under the same assumptions of Theorem~\ref{thm:QJS_is_sufficient}, it holds that 
 \begin{align}
  \overline{R}_{c}^{it}(D)=R_{c}^{it}(D).
 \end{align}
\finenunciado
\end{coro}

One important consequence of this result stems from the fact that, 
for a $\kappa$-th order Gauss Markov  stationary source and quadratic distortion, $R_{c}^{it}(D)$ can be found by solving a convex optimization problem over frequency-response magnitudes of linear-time invariant filters around an \textit{additive white-Gaussian noise} (AWGN) channel
~\cite[lemmas~3 and~5]{derost12}.

The operational relevance of Corollary~\ref{coro:overlineRcitD_equals_RcitD} is that 
when the latter AWGN channel is replaced by an entropy-coded  subtractively-dithered scalar quantizer, one obtains a source-coding scheme whose operational rate exceeds $R_{c}^{it}(D)$ by at most $0.254$ bits/sample
when operating causally and by at most 
$1.254$ bits/sample when operating with zero delay~\cite[section~VI]{derost12}.
Thanks to Corollary~\ref{coro:overlineRcitD_equals_RcitD}, it turns out that the operational data rate of such scheme lies within the same bounds with respect to $R_{c}^{it}(D)$ itself.


\section{Conclusions}\label{sec:Conclusions}
We have shown that, in general, the causal \textit{information rate-distortion function} (IRDF) $R_{c}^{it}(D)$ for one-sided stationary sources cannot be realized by a reconstruction which is jointly-stationary with the source.
Nevertheless, if the source is $\kappa$-th order Markovian, then
the search for the causal IRDF can be restricted to reconstructions which are jointly stationary with the source from the $\kappa$-th sample.
This led us to prove that $\overline{R}_{c}^{it}(D)$ actually coincides with $R_{c}^{it}(D)$ for a large class of distortion criteria.
This reveals that for Gauss-Markov sources and quadratic distortion, $R_{c}^{it}(D)$ can be found by solving the convex optimization problem derived in~\cite{derost12}.
It also implies that  for the same source and distortion,
a zero-delay average data rate exceeding 
$R_{c}^{it}(D)$ by not more than (approximately) $1.245$ bits/sample is achievable with the scheme proposed in~\cite{derost12}.


\section{Appendix}\label{sec:appendix}
%
\subsection{Proofs}\label{subsec:proofs}
\begin{proof}[Proof of Lemma~\ref{lem:Wsp_D_cannot_represent_asymptotic_SLDC}]\label{proof:of_lem_Wsp_D_cannot_represent_asymptotic_SLDC}
We will resort to a contradiction argument, and thus start
by supposing that there exists $(\Wsp_{D})=(\Asp_{D})$. 

Since $\rho$ is non-negative, there must exist 
a pair of random processes $(\rvax_{-\infty}^{\infty},\rvay_{-\infty}^{\infty})$ and a value $D\geq0$ for which~\eqref{eq:asymptotic_distortion_constraint0} holds with equality.
Hence, $(\rvax_{-\infty}^{\infty},\rvay_{-\infty}^{\infty})\in(\Asp_{D})$, and thus $(\rvax_{-\infty}^{\infty},\rvay_{-\infty}^{\infty})\in(\Wsp_{D})$.
From the definition of $(\Asp_{D})$, any other pair of processes 
$(\dot{\rvax}_{-\infty}^{\infty},\dot{\rvay}_{-\infty}^{\infty})$ 
which distributes exactly as 
$(\rvax_{-\infty}^{\infty},\rvay_{-\infty}^{\infty})$ everywhere except on a single positive index, say $\ell\in\Nl$, in which $P_{\dot{\rvax}(\ell),\dot{\rvay}(\ell)}\in\Pspa^{1,1}$ and
\begin{align}
 \rho(P_{\dot{\rvax}(\ell),\dot{\rvay}(\ell)})
 =D+1,
\end{align}
will also belong to $(\Asp_{D})$ and therefore 
$(\dot{\rvax}_{-\infty}^{\infty},\dot{\rvay}_{-\infty}^{\infty})\in(\Wsp_{D})$.
The latter means that, according to Definition~\ref{def:W_D}, there exists a pair of integers $\ell_{1},\ell_{2}$ with $\ell_{2}\in \Nl$ such that 
$\ell_{1}\leq \ell \leq \ell_{2}$ and 
\begin{align}
(\dot{\rvax}_{\ell_{1}}^{\ell_{2}},\dot{\rvay}_{\ell_{1}}^{\ell_{2}})\in
 ( \Wsp_{D}^{\ell_{1},\ell_{2}}).
\end{align}
This, together with the fact that the sets $\set{\Wsp_{D}^{t,s}}_{t\leq s\in\Z}$ satisfy~\eqref{eq:W_k,k_stationary_one_sided}, implies that 
\begin{align}
(\dot{\rvax}_{\ell_{1}}^{\ell_{2}},\dot{\rvay}_{\ell_{1}}^{\ell_{2}})
\in
( \Wsp_{D}^{\ell_{1}+mL,\ell_{2}+mL}), \fspace m\in\Nl,
\end{align}
where $L\eq \ell_{2}-\ell_{1}+1$.
Hence, any pair of random processes $(\ddot{\rvax}_{-\infty}^{\infty},\ddot{\rvay}_{-\infty}^{\infty})$ with pair-wise distributions given by 
\begin{align}
P_{\ddot{\rvax}(k),\ddot{\rvay}(k)}
=
\begin{cases}
 P_{\rvax(k),\rvay(k)}, & k\notin \set{t=\ell+mL:m\in\Nl } 
 \\
 P_{\dot{\rvax}(\ell),\dot{\rvay}(\ell)}, & k\in \set{t=\ell+mL:m\in\Nl }
\end{cases}
\end{align}
together with the collection of integers~$\set{k_{i}=\ell_{1}+iL:i\in\Nl}\cup \Z_{0}^{-}$ satisfy the conditions of Definition~\ref{def:W_D}, and thus
$(\ddot{\rvax}_{-\infty}^{\infty},\ddot{\rvay}_{-\infty}^{\infty})\in (\Wsp_{D})$.
However,
\begin{align}
  \lim_{n\to\infty}\frac{1}{n}\Sumfromto{k=1}{n}\rho\left(P_{\ddot{\rvax}(k),\ddot{\rvay}(k)} \right)=
  D+\frac{ 1 }{L}>D, 
\end{align}
meaning that $(\ddot{\rvax}_{-\infty}^{\infty},\ddot{\rvay}_{-\infty}^{\infty})\notin (\Asp_{D})$.
This contradicts the initial supposition that $(\Wsp_{D})=(\Asp_{D})$, completing the proof.
\end{proof}


\begin{proof}[Proof of Lemma~\ref{lem:liminf_leq_inflim}]\label{prof:lem_liminf_leq_inflim}
 From the definition of $\hat{R}_c^{it}(D)$  we have that
\begin{align}
\forall \e_1 > 0,\,\exists N_{\e_1} : 
\inf_{( {\rvax}_1^{n}, {\rvay}_1^{n})\in
(\Wsp_D^{1,n})\cap(\Csp^{n})
} 
\frac{1}{n} 
I ( {\rvax}_1^{n}; {\rvay}_1^{n})
\leq 
\hat{R}_c^{it}(D) + \e_1
,\fspace 
\forall n\geq N_{\e_{1}}.
\end{align}
By the definition of $\inf$, we have that
$\forall \e_2 > 0, n\geq N_{\e_{1}}$, there exists 
$(\dot{\rvax}_1^{n},\dot{\rvay}_1^{n})\in(\Wsp_{D}^{1,n})$ such that
\begin{align}\label{eq:breves_leq_RDF}
\frac{1}{n} I(\dot{\rvax}_1^{n},\dot{\rvay}_1^{n})
\leq 
\inf_{( {\rvax}_1^{n}, {\rvay}_1^{n})\in(\Wsp_D^{1,n})\cap(\Csp^{n})} \frac{1}{n} 
I ( {\rvax}_1^{n}; {\rvay}_1^{n})+ \e_2
\leq 
\hat{R}_c^{it}(D)  + \e_1 + \e_2.
\end{align}
Now consider the pair $(\breve{\rvax}_{1}^{\infty},\breve{\rvay}_{1}^{n})$ such that 
\begin{align}
 \breve{\rvax}_{1}^{\infty}&\sim \mathring{\rvax}_{1}^{\infty}\\
 P_{\breve{\rvay}_{1}^{n}|\breve{\rvax}_{1}^{\infty}}
 &=
 P_{\breve{\rvay}_{1}^{n}|\breve{\rvax}_{1}^{n}}
 =
 P_{\dot{\rvay}_{1}^{n}|\dot{\rvax}_{1}^{n}}
\end{align}
and build the processes $(\tilde{\rvax}_1^{\infty},\tilde{\rvay}_1^{\infty})$ as in Proposition~\ref{prop:construction}.
The mutual information rate between $\tilde{\rvax}_1^{kn}$ and $\tilde{\rvay}_1^{kn}$ can be upper bounded as 
\begin{align}
\frac{1}{kn} I(\tilde{\rvax}_1^{kn};\tilde{\rvay}_1^{kn})
\overset{\eqref{eq:Irate_of_block_less_than_sequence_of_blocks0}}{\leq}
\frac{1}{n} I(\breve{\rvax}_{1}^{n};\breve{\rvay}_{1}^{n})
\overset{\eqref{eq:breves_leq_RDF}}{\leq}
\hat{R}_c^{it}(D)  + \e_1 + \e_2.
\label{eq:ub_lem}
\end{align}

From~\eqref{eq:ub_lem} we also obtain that, for all $i\in\Nl$,
\begin{equation}
\begin{split}
\frac{1}{i}
 I(\tilde{\rvax}_1^{i};\tilde{\rvay}_1^{i})
& 
\overset{(\textbf{P1})}{\leq}
\frac{1}{i}
 I(\tilde{\rvax}_1^{\lceil i/n\rceil n};\tilde{\rvay}_1^{\lceil i/n\rceil n})
\overset{ \eqref{eq:ub_lem}}{\leq} 
 \frac{\left\lceil \frac{i}{n}\right\rceil n}{i}
 \left(\hat{R}_c^{it}(D)  + \e_1 + \e_2 \right)
 \\&
 \leq 
 \left(1+\frac{n}{i} \right)
 \left(\hat{R}_c^{it}(D)  + \e_1 + \e_2 \right).
\end{split} \label{eq:I_xtilde_1_i_leq}
\end{equation}
Recalling that the latter holds for all $i \in \Nl$ and that 
$(\tilde{\rvax}_{1}^{\infty}, \tilde{\rvay}_{1}^{\infty})\in(\Pspa_{D}^{\infty})$, we obtain 
\begin{align}
R_c^{it}(D)
\overset{\eqref{eq:RcitD_for_infty}}{=}
\inf_{( {\rvax}_1^{\infty}; {\rvay}_1^{\infty}) \in 
(\Pspa_D^{\infty})\cap (\Csp^{\infty})
}
\lim_{i \rightarrow \infty}
\frac{1}{i} I( {\rvax}_1^{i}; {\rvay}_1^{i})
\leq 
\bar{I}(\tilde{\rvax}_{1}^{\infty}; \tilde{\rvay}_{1}^{\infty})
\overset{\eqref{eq:I_xtilde_1_i_leq}}{\leq}
\hat{R}_c^{it}(D)  + \e_1 + \e_2.
\end{align}
Since this inequality is satisfied for all $\e_{1},\e_{2}>0$, it follows that 
$R_c^{it}(D) \leq \hat{R}_c^{it}(D)$, completing the proof.
\end{proof}

\begin{proof}[Proof of Lemma~\ref{lem:liminf_equals_inflim}]\label{proof:liminf_equals_inflim}
Since the conditions of Lemma~\ref{lem:liminf_leq_inflim} are satisfied, we have that 
$ R^{it}_{c}(D) \leq \hat{R}_{c}^{it}(D)$.
Therefore, it suffices to show that $ R^{it}_{c}(D) \geq \hat{R}_{c}^{it}(D)$.

We will first show that~\eqref{eq:weak_consistencia} 
$\Rightarrow$~$ R^{it}_{c}(D) \geq \hat{R}_{c}^{it}(D)$.
By the definition of $\inf$ on $R_c^{it}(D)$ we have that
\begin{align}
\forall \e_1 >0, \exists \,
(\breve{\rvax}_1^{\infty},\breve{\rvay}_1^{\infty}) \in \Pspa_D^{\infty}: 
\lim_{k \rightarrow \infty} \frac{1}{k} 
I(\breve{\rvax}_1^{k};\breve{\rvay}_1^{k}) 
\leq
R_c^{it}(D) + \e_1 . 
\end{align}
The latter means that for all $\e_2>0$, there exists a finite $N_{\e_2}$ such that
\begin{align}
\forall n \geq N_{\e_2}  : 
\frac{1}{n} I(\breve{\rvax}_1^{n};\breve{\rvay}_1^{n})
\leq
R_c^{it}(D) + \e_1 + \e_2.
\end{align}
Also, since $(\breve{\rvax}_1^{\infty},\breve{\rvay}_1^{\infty}) \in (\Pspa_D^{\infty})$, it follows from~\eqref{eq:weak_consistencia} that there exists a finite $N$ such that $(\breve{\rvax}_1^{n},\breve{\rvay}_1^{n}) \in (\Wsp_D^{1,n})$ for all $n \geq N$.
Since all the latter holds for all $n\geq \max\set{N_{\e_{2}},N}$, we obtain
\begin{align}
\lim_{n \rightarrow \infty} 
\inf_{({\rvax}_1^{n},{\rvay}_1^{n}) \in 
(\Wsp_D^{1,n})\cap(\Csp^{n})}
\frac{1}{n} I({\rvax}_1^{n};{\rvay}_1^{n})
\leq
R_c^{it}(D) + \e_1 + \e_2.
\end{align}
The latter is equivalent to
\begin{align}
\label{eq:ineq_2}
\hat{R}_c^{it}(D) \leq R_c^{it}(D) + \e_1 + \e_2. 
\end{align}
Since this inequality holds for all $\e_1,\e_2>0$, it follows that 
$ R^{it}_{c}(D) \geq \hat{R}_{c}^{it}(D)$, completing the first part of the proof.

We shall now prove that 
Assumption~\ref{assu:DFS_as_ineq_for_k} and
the continuity of $R_{c}^{it}(D)$ 
implies $R^{it}_{c}(D) \geq \hat{R}_{c}^{it}(D)$.
The continuity assumption on $R_{c}^{it}(D)$ means that  
\begin{align}
 \e_{\delta}
 \eq 
 R_{c}^{it}(D-\delta)
 - 
 R_{c}^{it}(D)
\end{align}
satisfies 
$
\lim_{\delta\to 0}\e_{\delta}=0
$,
for all $D>0$.
By the definition of $\inf$, we have that, for every $\delta>0,\e_{1}>0$, there exists a pair of processes 
$(\breve{\rvax}_{1}^{\infty},\breve{\rvay}_{1}^{\infty})\in 
(\Pspa_{D-\delta}^{\infty}) \cap(\Csp^{\infty})$
such that
\begin{align}
\lim_{k\to\infty}\frac{1}{k} I(\breve{\rvax}_1^k;\breve{\rvay}_1^k) \leq R_{c}^{it}(D-\delta) + \e_{1}
=R_{c}^{it}(D)
 +\e_{\delta}
 +\e_{1}
.
\end{align} 
The latter means that, for every $\e_{2}>0$, there exists a finite $N_{\e_{2}}$ such that, 
\begin{align}\label{eq:Ibreve_leq_Rcitd_plus_epsilons}
 \forall n\geq N_{\e_{2}}:\fspace
 \frac{1}{n}
 I(\breve{\rvax}_1^n;\breve{\rvay}_1^n) \leq 
 R_{c}^{it}(D)
 +\e_{\delta} + \e_{1}+\e_{2}.
\end{align}
Also, since $(\breve{\rvax}_{1}^{\infty},\breve{\rvay}_{1}^{\infty})\in (\Pspa_{D-\delta}^{\infty})$ and from the definition of $\Pspa_{D}^{\infty}$ in~\eqref{eq:Pdinfty_def}, it follows that there exists a finite $N_{\delta}$ such that 
\begin{align}
\forall n\geq N_{\delta}:\fspace  \rho^n(\breve{\rvax}_1^n,\breve{\rvay}_1^n) \leq D.
\end{align}
Thus, 
$ \forall \delta>0, \e_{2}>0$:
\begin{align}
\hat{R}^{it}(D)=
\lim_{n\to\infty}
 \inf_{
 (\rvax_{1}^{n},\rvay_{1}^{n})\in
 (\Wsp_{D}^{1,n})\cap(\Csp^{n})
 } \frac{1}{n}I(\rvax_{1}^{n};\rvay_{1}^{n})
 \leq 
 \bar{I}(\breve{\rvax}_1^\infty;\breve{\rvay}_1^\infty) \leq 
 R_{c}^{it}(D)
 +\e_{\delta} + \e_{1}+\e_{2}
\end{align}
Since this inequality holds for all $\delta,\e_{1},\e_{2}>0$, and recalling that $\e_{\delta}\to 0$ when $\delta\to 0$, it follows that  
$\hat{R}^{it}(D)\leq R^{it}(D)$, completing the proof.
\end{proof}

\begin{proof}[Proof of Theorem~\ref{thm:QJS_is_sufficient}]\label{proof:QJS_sufficient}
Since $\hat{R}^{it}(D)=R_{c}^{it}(D)$, it follows that for all $\e_{1}>0$, there exists a finite $N_{\e_{1}}$ such that 
\begin{align}
 \forall n\geq N_{\e_{1}}:
 \fspace 
 \inf_{(\rvax_{1}^{n},\rvay_{1}^{n})\in
 (\Wsp_{D}^{1,n})\cap(\Csp^{n})} \frac{1}{n}
 I(\rvax_{1}^{n};\rvay_{1}^{n})
 \leq 
 R_{c}^{it}(D)+\e_{1}.
\end{align}
Thus, for all $\e_{2}>0$ and for all $n\geq N_{\e_{1}}$ there exists a pair of sequences 
$(\breve{\rvax}_{1}^{n},\breve{\rvay}_{1}^{n})\in(\Wsp_{D}^{1,n})\cap(\Csp^{n})$ such that 
\begin{align}\label{eq:Ibreve_leq_Rcitd_plus_epsilons2}
 \frac{1}{n}I(\breve{\rvax}_{1}^{n};\breve{\rvay}_{1}^{n})\leq  R_{c}^{it}(D)+\e_{1}+\e_{2}.
\end{align}

The fact that  $\breve{\rvax}_{1}^{n}$ distributes as $\mathring{\rvax}_{1}^{n}$ allows one to define the stationary extension of $\breve{\rvax}_{1}^{n}$
such that 
\begin{align}
 \breve{\rvax}_{1}^{\infty}&\sim \mathring{\rvax}_{1}^{\infty}\\
 P_{\breve{\rvay}_{1}^{n}|\breve{\rvax}_{1}^{\infty}}
 &=
 P_{\breve{\rvay}_{1}^{n}|\breve{\rvax}_{1}^{n}}.
 \end{align}
The latter, together with the fact that $(\breve{\rvax}_{1}^{n},\breve{\rvay}_{1}^{n})\in (\Csp^{n})$
implies that 
$(\breve{\rvax}_{1}^{\infty},\breve{\rvay}_{1}^{n})$ satisfies~\eqref{eq:MC_causality_0} for $k=1,2,\ldots,n$.

Starting from $(\breve{\rvax}_{1}^{\infty},\breve{\rvay}_{1}^{n})$ we build 
the processes $(\tilde{\rvax}_1^{\infty},\tilde{\rvay}_1^{\infty})$ as in Proposition~\ref{prop:construction}. 
%
From this construction 
and the stationarity assumption on the distortion-feasible sets given by~\eqref{eq:W_D_are_stationary},
we have that 
\begin{align}\label{eq:tilde_nplus1_inPD}
\forall k\in\Nl_{0}:\fspace (\tilde{\rvax}_{kn+1}^{kn+n},\tilde{\rvay}_{kn+1}^{kn+n})\in (\Wsp_{D}^{kn+1,kn+n}),
\end{align}
and, from~\eqref{eq:Irate_of_block_less_than_sequence_of_blocks0}, that
\begin{align}
 \frac{1}{kn} I(\tilde{\rvax}_1^{kn};\tilde{\rvay}_1^{kn})
\leq
\frac{1}{n} I(\breve{\rvax}_{1}^{n};\breve{\rvay}_{1}^{n}).
\label{eq:ineq_integer_blocks}
\end{align}
Then, for any $m\in\Nl$, $ t\in\set{0,\ldots,n-1}$, if we define $\alpha \eq \frac{ t+m}{n}$, we have 
\begin{align}
I(\tilde{\rvax}_{ t+1}^{ t+m};\tilde{\rvay}_{ t+1}^{  t+m})
&
\overset{(\textbf{P\ref{property:Iabc_larger_than_Iab}})}{\leq} 
I(\tilde{\rvax}_{1}^{\lceil \alpha\rceil n }
 ;\tilde{\rvay}_{1}^{\lceil \alpha\rceil n})
\overset{\eqref{eq:ineq_integer_blocks}}{\leq} 
\lceil \alpha\rceil 
I(\breve{\rvax}_{1}^{n};\breve{\rvay}_{1}^{n})
\leq 
(\alpha +1)
I(\breve{\rvax}_{1}^{n};\breve{\rvay}_{1}^{n})
\\&
=\frac{ t+n}{n}I(\breve{\rvax}_{1}^{n};\breve{\rvay}_{1}^{n})
+
\frac{m}{n}I(\breve{\rvax}_{1}^{n};\breve{\rvay}_{1}^{n})
\leq 
2I(\breve{\rvax}_{1}^{n};\breve{\rvay}_{1}^{n})
+
\frac{m}{n}I(\breve{\rvax}_{1}^{n};\breve{\rvay}_{1}^{n}).
\end{align}
Dividing both sides by $m$ we obtain
\begin{align}
\label{eq:I_inequality}
 \forall t\in\set{0,1,\ldots,n-1}:
 \fspace 
 \frac{1}{m}I(\tilde{\rvax}_{t+1}^{t+m};\tilde{\rvay}_{t+1}^{t+m})
\leq 
\left( \frac{1}{n}
+
\frac{2}{m}
\right)
I(\breve{\rvax}_1^{n};\breve{\rvay}_1^n).
\end{align}

Now suppose that $\rvat$ is a random variable uniformly distributed over $\{0,1,\ldots,n-1\}$ and independent of $\tilde{\rvax}_1^{\infty}$.
Let $n\geq \kappa$ and 
define the pair of processes
$(\bar{\rvax}_{1}^{\infty},\bar{\rvay}_{1}^{\infty})$ as%
\footnote{A similar construction, for two-sided processes, was proposed in the proof of~\cite[Theorem~4]{gorpin73} (seemingly for the first time), building upon~\cite{ovseev68}.
The same idea was rediscovered in the proof of~\cite[Theorem~3.2]{kim-yh10}.
Here we adapt it for the case of one-sided processes.}
\begin{subequations}\label{subeq:xbar_ybar_def}
  \begin{align}
  \bar{\rvax}(k) &\eq \tilde{\rvax}(k+n+\rvat), \fspace k=-\kappa+2,-\kappa+3 ,\ldots \label{eq:barx_def}\\
  \bar{\rvay}(k) &\eq \tilde{\rvay}(k+n+\rvat), \fspace k=1,2,\ldots
\end{align}
\end{subequations}
It is easy to verify that 
$\bar{\rvax}_{-\kappa+2}^{\infty}\sim\rvax_{1}^{\infty}$,
$\bar{\rvax}_{-\kappa+2}^{\infty}\Perp\rvat$ (from the stationarity of $\rvax_{1}^{\infty}\sim\tilde{\rvax}_{1}^{\infty}$),
and that $(\bar{\rvax}_{1}^{\infty},\bar{\rvay}_{1}^{\infty})$ 
are jointly stationary.
Thus 
$\bar{\rvax}_{-\kappa+2}^{\infty}$
is $\kappa$-th order Markovian stationary.
These facts imply, in view of theorems~\ref{thm:conditions_forJS_and_causal} and~\ref{thm:stationarity_and_causality_Gaussian}, that
the pair $(\bar{\rvax}_{1}^{\infty},\bar{\rvay}_{1}^{\infty})$ may not be causally related according to~\eqref{eq:MC_causality_0}.
However, 
since 
$(\tilde{\rvax}_{1}^{\infty},\tilde{\rvay}_{1}^{\infty})$ satisfies~\eqref{eq:MC_causality_kappa}  (see Proposition~\ref{prop:construction}), we have that 
$(\bar{\rvax}_{1}^{\infty},\bar{\rvay}_{1}^{\infty})$ does satisfy the causality Markov chains 
\begin{align}\label{eq:MC_causality_bar}
  \bar{\rvax}_{k+1}^{\infty},
 \longleftrightarrow
 \bar{\rvax}_{-\kappa+2}^{k}
 \longleftrightarrow
 \bar{\rvay}_{1}^{k},
 \fspace k\in\Nl.
\end{align}

On the other hand,
in view of~\eqref{eq:tilde_nplus1_inPD} and thanks to Assumption~\ref{assu:Pspa_D_inf_assu}, we have that 
$(\tilde{\rvax}_{n+1}^{\infty},\tilde{\rvay}_{n+1}^{\infty})\in (\Pspa_{D}^{\infty})$.
Thus, from shift-invariance condition~\eqref{eq:shift_invariant_cond}, it readily follows that 
\begin{align}\label{eq:barxy_in_Pspa_D}
 (\bar{\rvax}_{1}^{\infty},\bar{\rvay}_{1}^{\infty})\in (\Pspa_{D}^{\infty}).
\end{align}
Now let 
\begin{align}
\label{eq:hatx}
 \hat{\rvax}_{1}^{\kappa-1}\eq \bar{\rvax}_{-\kappa+2}^{0}
\end{align}
and build 
$\hat{\rvay}_{1}^{\kappa-1}$ such that 
$(\hat{\rvax}_{1}^{\kappa-1},\hat{\rvay}_{1}^{\kappa-1})
\sim 
(\rvax_{1}^{\kappa-1},\dot{\rvay}_{1}^{\kappa-1})
$
and the Markov chain
\begin{align}\label{eq:MC_haty}
 \hat{\rvay}_{1}^{\kappa-1}
 \longleftrightarrow
 \hat{\rvax}_{1}^{\kappa-1}
 \longleftrightarrow
 (
 \bar{\rvax}_{1}^{\infty},\bar{\rvay}_{1}^{\infty},\rvat)
\end{align}
holds.
According to the ``first-samples condition'' in the statement of Theorem~\ref{thm:QJS_is_sufficient},
$
(\hat{\rvax}_{1}^{\kappa-1},\hat{\rvay}_{1}^{\kappa-1})\in(\Wsp_{D}^{1,\kappa-1})\cap(\Csp^{\kappa-1})
$
and $I(\hat{\rvax}_{1}^{\kappa-1};\hat{\rvay}_{1}^{\kappa-1}) < \infty$.
Now, concatenate 
$(\hat{\rvax}_{1}^{\kappa-1},\hat{\rvay}_{1}^{\kappa-1})$ with 
$(\bar{\rvax}_{1}^{\infty},\bar{\rvay}_{1}^{\infty})$ so as to obtain the pair of one-sided processes
\begin{subequations}\label{eq:ddot_x_and_ddot_y_def}
 \begin{align}
 \ddot{\rvax}_{1}^{\infty}
 \eq 
 \set{\hat{\rvax}(1),\hat{\rvax}(2),\ldots,\hat{\rvax}(\kappa-1),\bar{\rvax}(1),\bar{\rvax}(2),\ldots}\\
 \ddot{\rvay}_{1}^{\infty}
 \eq 
 \set{\hat{\rvay}(1),\hat{\rvay}(2),\ldots,\hat{\rvay}(\kappa-1),\bar{\rvay}(1),\bar{\rvay}(2),\ldots}
\end{align}
\end{subequations}
Since 
$(\hat{\rvax}_{1}^{\kappa-1},\hat{\rvay}_{1}^{\kappa-1})\in(\Wsp_{D}^{1,\kappa-1})$ and 
$(\bar{\rvax}_{1}^{\infty},\bar{\rvay}_{1}^{\infty})\in(\Pspa_{D}^{\infty})$ (see~\eqref{eq:barxy_in_Pspa_D}), it follows from Assumption~\ref{assu:Pspa_D_inf_assu} that 
\begin{align}
\label{eq:ddots_meets_distorsion}
 (\ddot{\rvax}_{1}^{\infty},\ddot{\rvay}_{1}^{\infty})\in(\Pspa_{D}^{\infty})
 \text{ and }
 (\ddot{\rvax}_{\kappa}^{\infty},\ddot{\rvay}_{\kappa}^{\infty})\in(\Pspa_{D}^{\infty}).
\end{align}
On the other hand,~\eqref{eq:MC_haty} and~\eqref{eq:MC_causality_bar} imply that 
%
  $\ddot{\rvax}_{k+1}^{\infty}
 \longleftrightarrow 
 \ddot{\rvax}_{1}^{k}
 \longleftrightarrow 
 \ddot{\rvay}_{1}^{k}$,
 $k\in\Nl$,
%
i.e, 
\begin{align}
\label{eq:ddots_are_causal}
(\ddot{\rvax}_{1}^{\infty},\ddot{\rvay}_{1}^{\infty})\in(\Csp^{\infty}).
\end{align}

The pair  $(\ddot{\rvax}_1^{\infty},\ddot{\rvay}_1^{\infty})$ further exhibits two important properties.
First, Lemma~\ref{lem:ddot_processes_are_kQJS} shows that the pairs of processes $(\ddot{\rvax}_1^{\infty},\ddot{\rvay}_1^{\infty})$ are $\kappa$-QJS (i.e. 
$(\ddot{\rvax}_1^{\infty},\ddot{\rvay}_1^{\infty}) \in (\Qsp_{\kappa})$). 
Second, as we show in Lemma~\ref{lem:MI_ddot_processes}, the mutual information rate of the pair processes 
$(\ddot{\rvax}_{1}^{\infty};\ddot{\rvay}_{1}^{\infty})$
lower bounds
$n^{-1}I(\breve{\rvax}_{1}^{n};\breve{\rvay}_{1}^{n})$, for all $n \in \Nl$.
Thus, for every 
$n \geq \max \set{N_{\e_1},\kappa}$, 
\begin{align*}
\lim_{m\to\infty}
  \frac{1}{m}
    I(\ddot{\rvax}_{1}^{m};\ddot{\rvay}_{1}^{m}) 
\leq
\frac{1}{n}
I(\breve{\rvax}_1^{n},\breve{\rvay}_1^n)
\overset{\eqref{eq:Ibreve_leq_Rcitd_plus_epsilons2}}{\leq}
R_{c}^{it}(D)
+\e_{1}+\e_{2}.
\end{align*}
Since the existence of a pair of processes
such as 
$(\ddot{\rvax}_{1}^{\infty},\ddot{\rvay}_{1}^{\infty})
\in 
(\Pspa_{D}^{\infty})
\cap 
(\Csp^{\infty})
\cap 
(\Qsp_{\kappa})
$ 
which also satisfy
$(\ddot{\rvax}_{\kappa}^{\infty},\ddot{\rvay}_{\kappa}^{\infty})
\in 
(\Pspa_{D}^{\infty})$
is guaranteed for every $\e_{1}>0$ and $\e_{2}>0$, it readily follows that the search for the infimum on the RHS of~\eqref{eq:RcitD_for_infty}
can be confined to such pairs, completing the proof.
\end{proof}

\begin{lem}
\label{lem:ddot_processes_are_kQJS}
Let $(\ddot{\rvax}_1^{\infty},\ddot{\rvay}_1^{\infty})$ be the concatenated processes defined in~\eqref{eq:ddot_x_and_ddot_y_def}, which are built from 
$(\bar{\rvax}_1^{\infty},\bar{\rvay}_1^{\infty})$ 
(see~\eqref{subeq:xbar_ybar_def})
and
$(\hat{\rvax}_1^{\kappa-1},\hat{\rvay}_1^{\kappa -1})$ 
(see~\eqref{eq:hatx} and~\eqref{eq:MC_haty} and the text between these equations)
and satisfy~\eqref{eq:ddots_are_causal} and~\eqref{eq:ddots_meets_distorsion}.
Then
$(\ddot{\rvax}_1^{\infty},\ddot{\rvay}_1^{\infty})\in (\Qsp_{\kappa}) $. 
\finenunciado
\end{lem}
\begin{proof}
%
By construction, $(\ddot{\rvax}_\kappa^{\infty},\ddot{\rvay}_\kappa^{\infty})$ are jointly stationary. 
Thus, 
all that remains to prove is that 
$ 
\bar{I}(\ddot{\rvax}_{\kappa}^{\infty};\ddot{\rvay}_{\kappa}^{\infty})
=
\bar{I}(\ddot{\rvax}_{1}^{\infty};\ddot{\rvay}_{1}^{\infty})
$ (see Definition~\ref{def:quasi_jointly_stat}).
For this purpose, notice that for all $i>\kappa$,
\begin{align}
   I(\ddot{\rvax}_{1}^{i};\ddot{\rvay}_{1}^{i})
   &
   \overset{\text{(cr)}}{=}
   I(\ddot{\rvax}_{1}^{i};\ddot{\rvay}_{\kappa}^{i})
   +
   I(\ddot{\rvax}_{1}^{i};\ddot{\rvay}_{1}^{\kappa-1}|\ddot{\rvay}_{\kappa}^{i})
   \\&
   \overset{\text{(cr)}}{=}
   I(\ddot{\rvax}_{1}^{i};\ddot{\rvay}_{\kappa}^{i})
   +
   I(\ddot{\rvax}_{1}^{i},\ddot{\rvay}_{\kappa}^{i};\ddot{\rvay}_{1}^{\kappa-1})
   -
   I(\ddot{\rvay}_{\kappa}^{i};\ddot{\rvay}_{1}^{\kappa-1})
   \\&
   \overset{(a)}{\leq}
   I(\ddot{\rvax}_{1}^{i};\ddot{\rvay}_{\kappa}^{i})
   +
   I(\ddot{\rvax}_{1}^{i},\ddot{\rvay}_{\kappa}^{i};\ddot{\rvay}_{1}^{\kappa-1})
   \\&
   \overset{\text{(cr)}}{=}
   I(\ddot{\rvax}_{1}^{i};\ddot{\rvay}_{\kappa}^{i})
   +
   I(\ddot{\rvax}_{1}^{\kappa-1};\ddot{\rvay}_{1}^{\kappa-1})
   +
   I(\ddot{\rvax}_{\kappa}^{i},\ddot{\rvay}_{\kappa}^{i} ; \ddot{\rvay}_{1}^{\kappa-1}|\ddot{\rvax}_{1}^{\kappa-1})
   \\&
   \overset{(b)}{=}
   I(\ddot{\rvax}_{1}^{i};\ddot{\rvay}_{\kappa}^{i})
   +
   I(\ddot{\rvax}_{1}^{\kappa-1};\ddot{\rvay}_{1}^{\kappa-1})
\\&
\overset{\text{(cr)}}{=}
  I(\ddot{\rvax}_{\kappa}^{i};\ddot{\rvay}_{\kappa}^{i})
  +
  I(\ddot{\rvay}_{\kappa}^{i}; \ddot{\rvax}_{1}^{\kappa-1} |\ddot{\rvax}_{\kappa}^{i})
  +
  I(\ddot{\rvay}_{1}^{\kappa-1};\ddot{\rvax}_{1}^{\kappa-1})
  \label{eq:three_MI} 
  \end{align}
where all the equalities labeled ``(cr)'' stem from the chain-rule of mutual information,
$(a)$ holds because $I(\rvaa;\rvab) \geq 0$ , and 
$(b)$ is a consequence of the fact that 
$\ddot{\rvax}_{1}^{\kappa-1}=\hat{\rvax}_{1}^{\kappa-1}$, 
$\ddot{\rvay}_{1}^{\kappa-1}=\hat{\rvay}_{1}^{\kappa-1}$,
and Markov chain~\eqref{eq:MC_haty}.
The mutual information in the middle of~\eqref{eq:three_MI} can be upper bounded as
\begin{align}
  I(\ddot{\rvay}_{\kappa}^{i}; \ddot{\rvax}_{1}^{\kappa-1} |\ddot{\rvax}_{\kappa}^{i})
  &
  \overset{(\text{cr})}{=}
  I(\ddot{\rvay}_{\kappa}^{i},\rvat; \ddot{\rvax}_{1}^{\kappa-1} |\ddot{\rvax}_{\kappa}^{i})
  -
  I(\rvat; \ddot{\rvax}_{1}^{\kappa-1} |\ddot{\rvax}_{\kappa}^{i},\ddot{\rvay}_{\kappa}^{i})
  \\
  &
  \overset{(a)}{\leq}
  I(\ddot{\rvay}_{\kappa}^{i},\rvat; \ddot{\rvax}_{1}^{\kappa-1} |\ddot{\rvax}_{\kappa}^{i})
  \\&
  \overset{(\text{cr})}{=}
  I(\ddot{\rvay}_{\kappa}^{i}; \ddot{\rvax}_{1}^{\kappa-1} |\ddot{\rvax}_{\kappa}^{i},\rvat)
  +
  I(\rvat; \ddot{\rvax}_{1}^{\kappa-1} |\ddot{\rvax}_{\kappa}^{i})
  \\
  &
  \overset{(b)}{=}
  I(\ddot{\rvay}_{\kappa}^{i}; \ddot{\rvax}_{1}^{\kappa-1} |\ddot{\rvax}_{\kappa}^{i},\rvat),
\label{eq:upper_bound}
\end{align}
where the equalities labeled $(\text{cr})$ are due to the chain rule of mutual information,
$(a)$ follows because mutual information is non-negative, and 
$(b)$ holds because $\bar{\rvax}_{-\kappa+2}^{\infty}\Perp\rvat$ implies 
$\ddot{\rvax}_{1}^{\infty}\Perp \rvat$.

Using the definition of $\ddot{\rvax}_1^{\infty},\ddot{\rvay}_1^{\infty}$ (see~\eqref{eq:ddot_x_and_ddot_y_def}) we have that, for the case $i\geq n+\kappa$,
\begin{align}
   I(\ddot{\rvay}_{\kappa}^{i}; \ddot{\rvax}_{1}^{\kappa-1} |\ddot{\rvax}_{\kappa}^{i},\rvat)
   &=
 I(\tilde{\rvay}_{n+1+\rvat}^{n+1+\rvat+i-\kappa}; 
 \tilde{\rvax}_{n+\rvat-\kappa+2}^{n+\rvat} 
 | \tilde{\rvax}_{n+1+\rvat}^{n+1+\rvat+i-\kappa},\rvat) 
 \\&
 \overset{(\textbf{P1})}{\leq}
  I(\tilde{\rvay}_{n+1}^{n+1+\rvat+i-\kappa}; 
 \tilde{\rvax}_{1}^{n+\rvat} 
 | \tilde{\rvax}_{n+1+\rvat}^{n+1+\rvat+i-\kappa},\rvat) 
\\&
 \overset{\hphantom{(a)}}{=}
  I(\rvab_{23},\rvab_{4}; 
 \rvaa_{1},\rvaa_{2} 
 | \rvaa_{3},\rvaa_{4},\rvat),\label{eq:first_steps}
 \end{align}
%
where the random elements
\begin{subequations}\label{eq:ab_defs}
\begin{align}
  \rvaa_{1}&\eq \tilde{\rvax}_{1}^{n}
  &
  \rvaa_{2}&\eq \tilde{\rvax}_{n+1}^{n+\rvat}
  \\
  \rvaa_{3}&\eq \tilde{\rvax}_{n+\rvat+1}^{2n}
  &
  \rvaa_{4}&\eq \tilde{\rvax}_{2n+1}^{ n+1+\rvat+i-\kappa}
  \\
    \rvab_{23}&\eq \tilde{\rvay}_{n+1}^{2n}
  &
    \rvab_{4}&\eq \tilde{\rvay}_{2n+1}^{ n+1+\rvat+i-\kappa}
  \end{align}
  \end{subequations}
are introduced so as to streamline the presentation of the following steps.
The relations between all these variables is illustrated in Fig.~\ref{fig:I_desarrollo}.  
Notice that for these sequences the Markov chains~\eqref{eq:MC_for_tilde_y} translate into
\begin{align}
  \rvab_{23} \longleftrightarrow&  \rvaa_{2},\rvaa_{3} \longleftrightarrow \rvaa_{1},\rvaa_{4},\rvab_{4} \label{eq:MCb23}\\
  \rvab_{4}  \longleftrightarrow& \rvaa_{4} \longleftrightarrow \rvaa_{1},\rvaa_{2},\rvaa_{3},\rvab_{23} \label{eq:MCb4}
\end{align}
With this, we can continue from~\eqref{eq:first_steps} and deduce that 
\begin{align}
    I(\ddot{\rvay}_{\kappa}^{i}; \ddot{\rvax}_{1}^{\kappa-1} |\ddot{\rvax}_{\kappa}^{i},\rvat)
   &\leq
  I(\rvab_{23},\rvab_{4};  \rvaa_{1},\rvaa_{2}  | \rvaa_{3},\rvaa_{4},\rvat)
  \\&
  \overset{(\text{cr})}{=}
  I(\rvab_{23};  \rvaa_{1},\rvaa_{2}  | \rvaa_{3},\rvaa_{4},\rvat)
  +
  I(\rvab_{4};  \rvaa_{1},\rvaa_{2}  |\rvab_{23}, \rvaa_{3},\rvaa_{4},\rvat)
  \\&
  \overset{\eqref{eq:MCb4}}{=}
  I(\rvab_{23};  \rvaa_{1},\rvaa_{2}  | \rvaa_{3},\rvaa_{4},\rvat)
  \\&
  \overset{(\text{cr})}{=}
  I(\rvab_{23};  \rvaa_{1},\rvaa_{2},\rvaa_{3},\rvaa_{4}  | \rvat)
  -
  I(\rvab_{23};  \rvaa_{3},\rvaa_{4} | \rvat)
  \\&
  \overset{(a)}{\leq}
  I(\rvab_{23};  \rvaa_{1},\rvaa_{2},\rvaa_{3},\rvaa_{4}  | \rvat)
  \\&
  \overset{(\text{cr})}{=}
    I(\rvab_{23};  \rvaa_{2},\rvaa_{3} | \rvat)
    +
    I(\rvab_{23};  \rvaa_{1},\rvaa_{4}  | \rvaa_{2},\rvaa_{3},\rvat)
 \\&
  \overset{\eqref{eq:MCb23}}{=}
    I(\rvab_{23};  \rvaa_{2},\rvaa_{3} | \rvat)
    \\&
    \overset{\eqref{eq:ab_defs}}{=}
  I(\tilde{\rvay}_{n+1}^{2n}; \tilde{\rvax}_{n+1}^{2n}| \rvat)
  \\&
  \overset{(b)}{=}
   I(\breve{\rvay}_{1}^{n}; \breve{\rvax}_{1}^{n})\label{eq:middle_I}
\end{align}
where $(a)$ follows because mutual information is non-negative and 
$(b)$ 
is due to~\eqref{eq:i_same0} and the fact that $(\tilde{\rvax}_1^{\infty},\tilde{\rvay}_1^{\infty}) \Perp \rvat$.
\begin{figure}[htbp]
  \centering
   \input{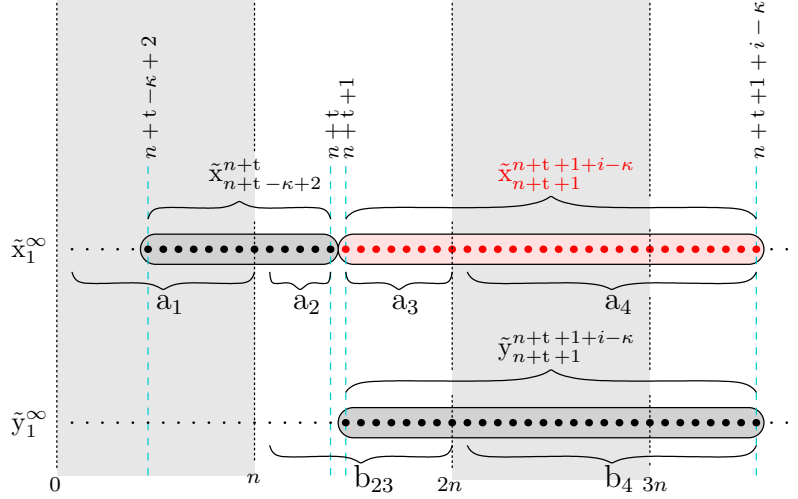}
\caption{Schematic representation of the change of variables introduced in~\eqref{eq:first_steps}.
Each dot represents one element in the sequences $\tilde{\rvax}_{1}^{\infty}$ and $\tilde{\rvay}_{1}^{\infty}$, with time increasing from left to right.}
\label{fig:I_desarrollo}
\end{figure}

Substituting~\eqref{eq:middle_I} into~\eqref{eq:upper_bound} and then the latter into~\eqref{eq:three_MI}, we arrive at 
\begin{align}
  I(\ddot{\rvax}_{\kappa}^{i};\ddot{\rvay}_{\kappa}^{i})
 \overset{\text{(P\ref{property:Iabc_larger_than_Iab})}}{\leq}
  I(\ddot{\rvax}_{1}^{i};\ddot{\rvay}_{1}^{i})
&
  \leq
  I(\ddot{\rvay}_{\kappa}^{i};\ddot{\rvax}_{\kappa}^{i})
  +
  I( \breve{\rvay}_{1}^{n}; \breve{\rvax}_{1}^{n})
  +
  I(\ddot{\rvay}_{1}^{\kappa-1};\ddot{\rvax}_{1}^{\kappa-1}).
\end{align}
Dividing by $i$ and taking the limit as $i\to\infty$, we conclude that, for all $n\in\Nl$,  
\begin{align}\label{eq:ddot_xy_is_kappa_QJS}
  \lim_{i\to\infty}
  \frac{1}{i}
    I(\ddot{\rvax}_{\kappa}^{i};\ddot{\rvay}_{\kappa}^{i})
    =
 \lim_{i\to\infty}
  \frac{1}{i}     I(\ddot{\rvax}_{1}^{i};\ddot{\rvay}_{1}^{i}),
\end{align}
proving the claim that 
$(\ddot{\rvax}_{1}^{\infty},\ddot{\rvay}_{1}^{\infty})\in (\Qsp_{\kappa})$.
\end{proof}

\begin{lem}
\label{lem:MI_ddot_processes}
Let $(\ddot{\rvax}_1^{\infty},\ddot{\rvay}_1^{\infty})$ the concatenated processes defined in~\eqref{eq:ddot_x_and_ddot_y_def}, built from 
$(\bar{\rvax}_1^{\infty},\bar{\rvay}_1^{\infty})$ 
(see~\eqref{subeq:xbar_ybar_def})
and
$(\hat{\rvax}_1^{\kappa-1},\hat{\rvay}_1^{\kappa -1})$ 
(see~\eqref{eq:hatx} and~\eqref{eq:MC_haty} and the text between these equations).
Then
%
\begin{align*}
\lim_{m \to \infty} 
\frac{1}{m}
I(\ddot{\rvax}_1^{m};\ddot{\rvay}_1^{m})
\leq 
\frac{1}{n} I(\breve{\rvax}_1^{n};\breve{\rvay}_1^{n}).
\end{align*}
\finenunciado
\end{lem}
\begin{proof}
First, notice that~\eqref{eq:ddot_xy_is_kappa_QJS},~\eqref{subeq:xbar_ybar_def} and~\eqref{eq:ddot_x_and_ddot_y_def} imply 
\begin{align}\label{eq:ibar_ddot_equals_ibar_bar}
\bar{I}(\ddot{\rvax}_{1}^{\infty};\ddot{\rvay}_{1}^{\infty})
=
\bar{I}(\ddot{\rvax}_{\kappa}^{\infty};\ddot{\rvay}_{\kappa}^{\infty})
=
 \bar{I}(\bar{\rvax}_{1}^{\infty};\bar{\rvay}_{1}^{\infty}).
\end{align}
On the other hand
\begin{align}
I(\bar{\rvax}_{1}^{m};\bar{\rvay}_{1}^{m})
&
\overset{(a)}{\leq} 
I(\bar{\rvax}_{1}^{m};\bar{\rvay}_{1}^{m})
+
I(\bar{\rvax}_{1}^{m};\rvat|\bar{\rvay}_{1}^{m})
\\&
\overset{\text{(cr)}}{=}
I(\bar{\rvax}_{1}^{m};\bar{\rvay}_{1}^{m},\rvat)
\\&
\overset{(b)}{=}
I(\bar{\rvax}_{1}^{m};\bar{\rvay}_{1}^{m},\rvat)
-
I(\bar{\rvax}_{1}^{m};\rvat)
\\&
\overset{\text{(cr)}}{=}
I(\bar{\rvax}_{1}^{m};\bar{\rvay}_{1}^{m} |\rvat),
\label{eq:I_xybar_upperbound}
\end{align}
where 
$(a)$ is due to the non-negativity of mutual information and 
$(b)$ holds because $\rvat\Perp \bar{\rvax}_{1}^{\infty}$.
But
\begin{align}
 \frac{1}{m}
 I(\bar{\rvax}_{1}^{m};\bar{\rvay}_{1}^{m}|\rvat)
 &
 \overset{\eqref{subeq:xbar_ybar_def}}{=}
\frac{1}{m}
 I(\tilde{\rvax}_{n+\rvat+1}^{n+\rvat+m};\tilde{\rvay}_{n+\rvat+1}^{n+\rvat+m}|\rvat)
\\&
 \overset{(a)}{=}
 \frac{1}{m}
 I(\tilde{\rvax}_{\rvat+1}^{\rvat+m};\tilde{\rvay}_{\rvat+1}^{\rvat+m}|\rvat)
 \\&
  \overset{\eqref{eq:I_inequality}}{\leq}
 \frac{1}{n}I(\breve{\rvax}_{1}^{n};\breve{\rvay}_{1}^{n})
 + 
 \frac{2}{m} I(\breve{\rvax}_{1}^{n};\breve{\rvay}_{1}^{n}),
 \label{eq:bound1}
\end{align}
where 
$(a)$ follows from the fact that, for all $t\in\set{0,1,\ldots,n-1}$,
$(\tilde{\rvax}_{n+t+1}^{n+t+m},\tilde{\rvay}_{n+t+1}^{n+t+m})
\sim
(\tilde{\rvax}_{t+1}^{t+m},\tilde{\rvay}_{t+1}^{t+m})$
(see~\eqref{subeq:tilde_xy_def}). 
Substituting this into~\eqref{eq:I_xybar_upperbound} we obtain that
for every $n\in\Nl$ and $m\in\Nl$,
\begin{align}\label{eq:lalalala}
\frac{1}{m}
  I(\bar{\rvax}_{1}^{m};\bar{\rvay}_{1}^{m}) 
  &
  \leq 
  \frac{1}{n}I(\breve{\rvax}_{1}^{n};\breve{\rvay}_{1}^{n})
  +
   \frac{2}{m}I(\breve{\rvax}_{1}^{n};\breve{\rvay}_1^{n}).
  \end{align}
The proof is completed by taking the limit as $m\to\infty$ and substituting~\eqref{eq:ibar_ddot_equals_ibar_bar} in it.
\end{proof} 

\begin{proof}[Proof of Proposition~\ref{prop:there_exits_ybar_strongly_causal}]\label{proof:of_prop:there_exits_ybar_strongly_causal}
From Fact~\ref{fact:abc_and_barabc}, there exists $\bar{\rvay}_{1}^{\infty}$ such that \eqref{eq:likewise} holds and which satisfies 
\begin{align}
  \rvax_{-\infty}^{0}
  \longleftrightarrow 
 \rvax_{1}^{\infty} 
  \longleftrightarrow
   \bar{\rvay}_{1}^{\infty}.  \label{eq:conditional_past_independence}
\end{align}
Combining~\eqref{eq:likewise} with~\eqref{eq:short_causality} one obtains
\begin{align}\label{eq:esta1}
\underbrace{\rvax_{k+1}^{\infty}}_{\rvad}
  \longleftrightarrow
  \underbrace{\rvax_{1}^{k}}_{\rvac}
  \longleftrightarrow
  \underbrace{\bar{\rvay}_{1}^{k}}_{\rvaa}
  ,\fspace k\in\Nl.
\end{align}
On the other hand,~\eqref{eq:conditional_past_independence} readily  implies that
\begin{align}\label{eq:esta2}
\underbrace{\rvax_{-\infty}^{0}}_{\rvab}
  \longleftrightarrow 
    \underbrace{\rvax_{1}^{\infty}}_{\rvac,\rvad}
  \longleftrightarrow
  \underbrace{\bar{\rvay}_{1}^{k}}_{\rvaa}.
\end{align}
Applying Proposition~\ref{prop:abcd} with $\rvaa,\rvab,\rvac,\rvad$ corresponding to the labels placed under the terms in~\eqref{eq:esta1} and~\eqref{eq:esta2}, we obtain directly~\eqref{eq:strong_causality_ybar},  completing the proof. 
\end{proof}

\begin{proof}[Proof of Corollary~\ref{coro:overlineRcitD_equals_RcitD}]\label{proof:coro_overlineRcitD_equals_RcitD}
We will first show that $  \overline{R}_{c}^{it}(D)\leq R_{c}^{it}(D)$ and then that the reverse inequality is true as well.

From Theorem~\ref{thm:QJS_is_sufficient_strong}, 
$R_{c}^{it}(D) = R_{c\text{(strong)}}^{it}(D)$.
Also, Theorem~\ref{thm:QJS_is_sufficient_strong} states that the search for $R_{c\text{(strong)}}^{it}(D)$ can be confined to pairs of processes 
 $({\rvax}_{-\infty}^{\infty},{\rvay}_{1}^{\infty})$ which satisfy~\eqref{eq:MC_causality_hybrid} and
 such that   
 $
 ({\rvax}_{1}^{\infty},{\rvay}_{1}^{\infty})\in
 (\Qsp_{\kappa})\cap (\Pspa_{D}^{\infty}))$
 and 
 $
 ({\rvax}_{\kappa}^{\infty},{\rvay}_{\kappa}^{\infty})\in
 (\Pspa_{D}^{\infty})
$.
For each such pair, one can construct the pair 
of processes
 $ (\ddot{\rvax}_{-\infty}^{\infty},\ddot{\rvay}_{1}^{\infty})$ as
 \begin{subequations}\label{eq:ddot_construction}
 \begin{align}
 \ddot{\rvax}_{-\infty}^{0} &\eq {\rvax}_{-\infty}^{\kappa-1} \\
 (\ddot{\rvax}_{1}^{\infty},\ddot{\rvay}_{1}^{\infty})&
 \eq ({\rvax}_{\kappa}^{\infty},{\rvay}_{\kappa}^{\infty}).
 \end{align}
  \end{subequations}
 This construction yields that 
 \begin{align}
(\ddot{\rvax}_{1}^{\infty},\ddot{\rvay}_{1}^{\infty})&\in (\Pspa_{D}^{\infty})  
\\
(\ddot{\rvax}_{1}^{\infty},\ddot{\rvay}_{1}^{\infty})&\in (\Qsp_{1}) 
\text{  (since  $({\rvax}_{1}^{\infty},{\rvay}_{1}^{\infty})\in (\Qsp_{\kappa})$)}
 \\ 
 \bar{I}(\ddot{\rvax}_{-\kappa+2}^{\infty};\ddot{\rvay}_{1}^{\infty})
 &\overset{\eqref{eq:ddot_construction}}{=}
 \bar{I}({\rvax}_{1}^{\infty};{\rvay}_{\kappa}^{\infty})
 = 
 \bar{I}({\rvax}_{1}^{\infty};{\rvay}_{1}^{\infty}),\label{eq:estasi} 
 \end{align}
 where the last equality stems from the fact that 
 $\bar{I}({\rvax}_{1}^{\infty};{\rvay}_{1}^{\infty})
 \geq 
 \bar{I}({\rvax}_{1}^{\infty};{\rvay}_{\kappa}^{\infty})
 \geq 
 \bar{I}({\rvax}_{\kappa}^{\infty};{\rvay}_{\kappa}^{\infty})$
 and recalling that 
 $({\rvax}_{1}^{\infty},{\rvay}_{1}^{\infty})\in (\Qsp_{\kappa})$ 
 and the definition of $(\Qsp_{\kappa})$
 implies  
 $
 \bar{I}({\rvax}_{1}^{\infty};{\rvay}_{1}^{\infty})
 =
 \bar{I}({\rvax}_{\kappa}^{\infty};{\rvay}_{\kappa}^{\infty})
 $.
In addition, the fact that $(\rvax_{-\infty}^{\infty},\rvay_{1}^{\infty})$ satisfies~\eqref{eq:MC_causality_hybrid} implies that 
\begin{align}
(\ddot{\rvax}_{-\infty}^{-\kappa+1}, \ddot{\rvax}_{k+1}^{\infty})
 \longleftrightarrow 
 \ddot{\rvax}_{-\kappa+2}^{k}
\longleftrightarrow 
\ddot{\rvay}_{1}^{k} 
\end{align}
which in turn leads directly to 
\begin{align}
\ddot{\rvax}_{k+1}^{\infty}
 \longleftrightarrow 
 \ddot{\rvax}_{-\infty}^{k}
\longleftrightarrow 
\ddot{\rvay}_{1}^{k} 
\end{align}
and to
\begin{align}
 \bar{I} (\ddot{\rvax}_{\ell}^{n};\ddot{\rvay}_{1}^{n}) = 
 \bar{I}(\ddot{\rvax}_{-\kappa+2}^{n};\ddot{\rvay}_{1}^{n}), \fspace \forall \ell <-\kappa+2, \forall n\in\Nl.
\end{align}
This leads to 
\begin{align}
 \lim_{\ell\to-\infty}
 \lim_{n\to\infty}
 \frac{1}{n}
\bar{I}(\ddot{\rvax}_{\ell}^{n};\ddot{\rvay}_{1}^{n})
 =
  \lim_{n\to\infty}
 \frac{1}{n}
I(\ddot{\rvax}_{-\kappa+2}^{n};\ddot{\rvay}_{1}^{n})
=
\bar{I}(\ddot{\rvax}_{-\kappa+2}^{\infty};\ddot{\rvay}_{1}^{\infty})
\overset{\eqref{eq:estasi}}{=}
\bar{I}(\rvax_{1}^{\infty};\rvay_{1}^{\infty}).
\end{align}

Therefore, for every pair of processes which satisfies the constraints associated with $R_{c}^{it}(D)$, there exists another pair which satisfies the constraints in the definition of $\overline{R}_{c}^{it}(D)$ and yields the same information rate.
This proves that $\overline{R}_{c}^{it}(D)\leq R_{c}^{it}(D)$.

In order to show that $\overline{R}_{c}^{it}(D)\geq R_{c}^{it}(D)$, consider any pair 
$
(\rvax_{-\infty}^{\infty},\rvay_{1}^{\infty})$ 
satisfying~\eqref{eq:MC_causality_2sdd}
and such that 
$
(\rvax_{1}^{\infty},\rvay_{1}^{\infty})\in(\Qsp_{1})\cap(\Pspa_{D}^{\infty})
$.
Construct the pair of processes 
$(\dot{\rvax}_{1}^{\infty},\dot{\rvay}_{1}^{\infty})$ as
\begin{subequations}\label{eq:dot_rvars_def}
\begin{align}
 (\dot{\rvax}_{\kappa}^{\infty},\dot{\rvay}_{\kappa}^{\infty})
 &\eq
 (\rvax_{1}^{\infty},\rvay_{1}^{\infty})
 \\
 \dot{\rvax}_{1}^{\kappa-1}&\eq \rvax_{-\kappa+2}^{0},
\end{align}
\end{subequations}
and let the joint distribution of the pair 
$(\dot{\rvax}_{1}^{\kappa-1},\dot{\rvay}_{1}^{\kappa-1})$
be such that 
$(\dot{\rvax}_{1}^{\kappa-1},\dot{\rvay}_{1}^{\kappa-1})\in(\Wsp_{D}^{1,\kappa-1})\cap(\Csp^{\kappa-1})$ and 
\begin{align}\label{eq:MC_first_piece}
\dot{\rvay}_{1}^{\kappa-1}
\longleftrightarrow 
\dot{\rvax}_{1}^{\kappa-1}
\longleftrightarrow 
\dot{\rvax}_{\kappa}^{\infty},\dot{\rvay}_{\kappa}^{\infty}
\end{align}
(the existence of such pair is guaranteed by the ``first-samples condition'' in the statement of Theorem~\ref{thm:QJS_is_sufficient}).
From Assumption~\ref{assu:Pspa_D_inf_assu}, this construction yields 
\begin{align}
 (\dot{\rvax}_{1}^{\infty},\dot{\rvay}_{1}^{\infty})
 \in 
 (\Pspa_{D}^{\infty}).
\end{align}
The fact that 
$
(\rvax_{-\infty}^{\infty},\rvay_{1}^{\infty})$ 
satisfies~\eqref{eq:MC_causality_2sdd} translates into 
\begin{align}\label{eq:do}
\underbrace{\dot{\rvay}_{\kappa}^{k}}_{\rvab}
 \longleftrightarrow 
 (\underbrace{\dot{\rvax}_{-\infty}^{0}}_{\rvad}
, \underbrace{\dot{\rvax}_{1}^{k}}_{\rvac}
)
 \longleftrightarrow 
\underbrace{ \dot{\rvax}_{k+1}^{\infty}}_{\rvaa}
,\fspace k=\kappa,\kappa+1,\ldots,
\end{align}
The $\kappa$-th order Markovianity of $\dot{\rvax}_{-\infty}^{\infty}$ yields that  
\begin{align}\label{eq:re}
 \underbrace{\dot{\rvax}_{-\infty}^{0}}_{\rvad}
\longleftrightarrow
 \underbrace{\dot{\rvax}_{1}^{k}}_{\rvac}
 \longleftrightarrow 
\underbrace{ \dot{\rvax}_{k+1}^{\infty}}_{\rvaa},
\fspace k=\kappa,\kappa+1,\ldots,
\end{align}
Applying Proposition~\ref{prop:abcd} to~\eqref{eq:do} and~\eqref{eq:re} with variables in the proposition assigned according to the labels under~\eqref{eq:do} and~\eqref{eq:re}, we obtain that 
\begin{align}\label{eq:mi}
\dot{\rvay}_{\kappa}^{k} 
 \longleftrightarrow 
 \dot{\rvax}_{1}^{k} 
 \longleftrightarrow 
\dot{\rvax}_{k+1}^{\infty} 
,\fspace k=\kappa,\kappa+1,\ldots,
\end{align}
which combined with 
$(\dot{\rvax}_{1}^{\kappa-1},\dot{\rvay}_{1}^{\kappa-1})\in(\Wsp_{D}^{1,\kappa-1})\cap(\Csp^{\kappa-1})$
yields
$(\dot{\rvax}_{1}^{\infty},\dot{\rvay}_{1}^{\infty})\in(\Pspa_{D}^{\infty})\cap(\Csp^{\infty})$.
In addition, 
\begin{align}
  I(\dot{\rvax}_{1}^{n};\dot{\rvay}_{1}^{n})
  &
  \overset{(\text{cr})}{=}
  I(\dot{\rvax}_{1}^{n};\dot{\rvay}_{\kappa}^{n})
  +
  I(\dot{\rvax}_{1}^{n};\dot{\rvay}_{1}^{\kappa-1}|\dot{\rvay}_{\kappa}^{n})
\\&
  \overset{(\textbf{P\ref{property:Iabc_larger_than_Iab_gvn_c}})}{\leq}
  I(\dot{\rvax}_{1}^{n};\dot{\rvay}_{\kappa}^{n})
  +
  I(\dot{\rvax}_{1}^{n},\dot{\rvay}_{\kappa}^{n};\dot{\rvay}_{1}^{\kappa-1})
\\&
  \overset{(\text{cr})}{=}
  I(\dot{\rvax}_{1}^{n};\dot{\rvay}_{\kappa}^{n})
  +
  I(\dot{\rvax}_{1}^{n};\dot{\rvay}_{1}^{\kappa-1})
  +
    I(\dot{\rvay}_{\kappa}^{n};\dot{\rvay}_{1}^{\kappa-1}|\dot{\rvax}_{1}^{n})
\\&
\overset{\eqref{eq:MC_first_piece}}{=}
  I(\dot{\rvax}_{1}^{n};\dot{\rvay}_{\kappa}^{n})
  +
  I(\dot{\rvax}_{1}^{\kappa-1};\dot{\rvay}_{1}^{\kappa-1}).
  \end{align}
On the other hand,
\begin{align}
  I(\dot{\rvax}_{1}^{n};\dot{\rvay}_{\kappa}^{n})
  \overset{\eqref{eq:dot_rvars_def}}{=}
  I( \rvax_{-\kappa+2}^{n-\kappa+1}; \rvay_{1}^{n-\kappa+1})
  \overset{(\textbf{P\ref{property:Iabc_larger_than_Iab})}}{\leq} 
   I( \rvax_{\ell}^{n-\kappa+1}; \rvay_{1}^{n-\kappa+1})
\end{align}
for all $\ell< -\kappa+2$.
Thus
\begin{align}
  \lim_{n\to\infty} \frac{1}{n}  I(\dot{\rvax}_{1}^{n};\dot{\rvay}_{1}^{n})
  &\leq
  \lim_{\ell\to-\infty}
  \lim_{n\to\infty}\frac{1}{n} 
  \left[
      I( \rvax_{\ell}^{n-\kappa+1}; \rvay_{1}^{n-\kappa+1})
      + I(\dot{\rvax}_{1}^{\kappa-1};\dot{\rvay}_{1}^{\kappa-1})
      \right]
     \\& 
     =
       \lim_{\ell\to-\infty} \lim_{k\to\infty}\frac{1}{k} 
      I( \rvax_{\ell}^{k}; \rvay_{1}^{k}).\label{eq:limits_ineq}
\end{align}
Hence, 
for every pair $(\rvax_{-\infty}^{\infty},\rvay_{1}^{\infty})\in (\Csp_{-\infty}^{\infty})$ such that 
$
(\rvax_{1}^{\infty},\rvay_{1}^{\infty})\in(\Qsp_{1})\cap(\Pspa_{D}^{\infty})
$
there exists a pair 
$(\dot{\rvax}_{1}^{\infty},\dot{\rvay}_{1}^{\infty})\in(\Pspa_{D}^{\infty})\cap(\Csp^{\infty})$
satisfying~\eqref{eq:limits_ineq}.
This readily implies that $R_{c}^{it}(D)\leq \overline{R}_{c}^{it}(D)$, completing the proof.
\end{proof}


\subsection{Other Technical Results}
\label{sec:Other_Technical_Results}

\begin{prop}
\label{prop:I_for_MC}
For any random elements $\rvaa_{1},\rvaa_{2},\rvab_{1},\rvab_{2}$ 
satisfying the Markov chains
\begin{align}
  \rvaa_{2},\rvab_{2}
  & \longleftrightarrow
  \rvaa_{1}
  \longleftrightarrow
  \rvab_{1}
  \label{eq:a1_and_b1}
  \\
  \rvaa_{1},\rvab_{1}
  & \longleftrightarrow
  \rvaa_{2}
  \longleftrightarrow
  \rvab_{2}
  \label{eq:a2_and_b2}
\end{align}
it holds that 
\begin{align}
I(\rvaa_{1},\rvaa_{2};\rvab_{1},\rvab_{2}) 
=
I(\rvaa_{1};\rvab_{1})
	+
I(\rvaa_{2};\rvab_{2})
    -
I(\rvab_{1};\rvab_{2})
\end{align}
\finenunciado
\end{prop}

\begin{proof}[Proof of Proposition~\ref{prop:I_for_MC}]
The mutual information between $\rvaa_{1},\rvaa_{2}$ and $\rvab_{1},\rvab_{2}$ is given by
\begin{align}
    I(\rvaa_{1},\rvaa_{2};\rvab_{1},\rvab_{2})
    &
    \overset{\text{(cr)}}{=}
    I(\rvaa_{1},\rvaa_{2};\rvab_{1})
    +
    I(\rvaa_{1},\rvaa_{2};\rvab_{2}|\rvab_{1})
    \\&
    \overset{\text{(cr)}}{=}
    I(\rvaa_{1};\rvab_{1})
    +
    I(\rvaa_{2};\rvab_{1}|\rvaa_{1})
    +
    I(\rvaa_{1},\rvaa_{2};\rvab_{2}|\rvab_{1})
    \\&
    \overset{\eqref{eq:a1_and_b1}}{=}
    I(\rvaa_{1};\rvab_{1})
    +
    I(\rvaa_{1},\rvaa_{2};\rvab_{2}|\rvab_{1})
    \\&
    \overset{\text{(cr)}}{=}
    I(\rvaa_{1};\rvab_{1})
    +
    I(\rvaa_{1},\rvaa_{2},\rvab_{1};\rvab_{2})
    -
    I(\rvab_{1};\rvab_{2})
\\&
    \overset{\text{(cr)}}{=}
    I(\rvaa_{1};\rvab_{1})
    +
    I(\rvaa_{2};\rvab_{2})
    +
    I(\rvaa_{1},\rvab_{1};\rvab_{2}|\rvaa_{2})
    -
    I(\rvab_{1};\rvab_{2})
\\&
    \overset{\eqref{eq:a2_and_b2}}{=}
    I(\rvaa_{1};\rvab_{1})
    +
    I(\rvaa_{2};\rvab_{2})
    -
    I(\rvab_{1};\rvab_{2}),
    \end{align}
where all equalities labeled $\text{(cr)}$ stem from the chain rule of mutual information. 
\end{proof}

\begin{prop}\label{prop:abcd}
Let $\rvaa,\rvab,\rvac,\rvad$.
Then
  \begin{align}\label{eq:prop_abcd_claim2}
 \set{\rvaa 
  \longleftrightarrow
  \rvac
  \longleftrightarrow
  \rvad\fspace 
  \wedge
  \fspace 
   \rvaa 
  \longleftrightarrow
  \rvac,\rvad 
  \longleftrightarrow
  \rvab}
  \fspace 
  \iff
  \fspace
  \rvaa 
  \longleftrightarrow
  \rvac
  \longleftrightarrow
  (\rvab,\rvad).
  \end{align}
 \finenunciado
\end{prop}

\begin{proof}
The chain rule of mutual information yields
 \begin{align}
  I(\rvaa;\rvab,\rvad|\rvac)
=
  I(\rvaa;\rvad|\rvac)
  +
  I(\rvaa;\rvab|\rvac,\rvad).
    \label{eq:cadena}
 \end{align}
Since mutual information is non-negative, it follows that 
$  I(\rvaa;\rvab,\rvad|\rvac)=0$ if and only if both 
$  I(\rvaa;\rvad|\rvac)$ and
$ I(\rvaa;\rvab|\rvac,\rvad)$ are zero.
The proof is completed by noting that the statement $I(\rvau;\rvaw|\rvaz)=0$ is equivalent to the Markov chain 
$\rvau \longleftrightarrow \rvaz\longleftrightarrow \rvaw$.
\end{proof}



\end{document}